\documentclass[twocolumn]{IEEEtran}
\usepackage[final]{graphicx}
\usepackage{amsthm}
\usepackage{amsmath,epsfig,amssymb,verbatim,amsopn,cite,multirow}
\usepackage{balance}
\usepackage{multirow}
\usepackage{footnote}
\usepackage{algorithm}
\usepackage{algpseudocode}
\usepackage[usenames,dvipsnames]{color}
\usepackage[all]{xy}  
\usepackage{url}
\usepackage{subcaption}
\usepackage{graphicx}
\usepackage{array}

\newtheorem{Remark}{Remark}

  {\proof}{\proofend}
\newtheorem{proposition}{Proposition}

\newcommand{\qa}{{\bf a}}

\newcommand{\qd}{{\bf d}}

\newcommand{\qg}{{\bf g}}
\newcommand{\qh}{{\bf h}}

\newcommand{\qp}{{\bf p}}
\newcommand{\qq}{{\bf q}}

\newcommand{\qs}{{\bf s}}

\newcommand{\qu}{{\bf u}}

\newcommand{\qB}{{\bf B}}

\newcommand{\qF}{{\bf F}}
\newcommand{\qG}{{\bf G}}

\newcommand{\qI}{{\bf I}}

\newcommand{\qP}{{\bf P}}
\newcommand{\qQ}{{\bf Q}}

\newcommand{\Kmax}{K_{\mathrm{max}}}

\newcommand{\Csm}{\mathcal{C}_s}
\newcommand{\Usm}{\mathcal{U}_s}

\newcommand{\Kset}{ {\mathcal {K}}}
\newcommand{\Mset}{ {\mathcal {M}}}
\newcommand{\Qset}{ {\qq}}

\newcommand{\Pb}{P}
\newcommand{\Sn}{\sigma^2}

\newcommand{\Pu}{\varrho_\mathrm{u}}

\newcommand{\Unsm}{\mathcal{U}_{-s}}
\newcommand{\Csnm}{\mathcal{C}_{-s}}
\newcommand{\aset}{\qa_s}
\newcommand{\qbar}{\bar{\qq}}

\newcommand{\tauu}{\tau_{\mathrm{u}}}

\newcommand{\SINRk}{\mathrm{SINR}_{k}}
\newcommand{\SLINRk}{\mathrm{SLINR}_{k}}

\newcommand{\Ex}{\mathbb{E}}
\newcommand{\opt}{\mathrm{opt}}

\title{\fontsize{0.83cm}{1cm}\selectfont   Cluster-Wise Processing in Fronthaul-Aware Cell-Free Massive MIMO Systems}
\author{Zahra Mobini,~\IEEEmembership{Senior Member,~IEEE,} Ahmet Hasim Gokceoglu,  Li Wang,~\IEEEmembership{Senior Member,~IEEE,} Gunnar Peters,\\ Hyundong~Shin,~\IEEEmembership{Fellow,~IEEE}, and Hien Quoc Ngo,~\IEEEmembership{Fellow,~IEEE}}

\begin{document}

\bstctlcite{IEEEexample:BSTcontrol}
\maketitle

\begin{abstract} 
We exploit a general cluster-based network architecture for a fronthaul-limited  user-centric cell-free massive multiple-input multiple-output (CF-mMIMO) system under different degrees of cooperation among the access points (APs) to achieve scalable implementation. In particular, we consider a CF-mMIMO system wherein the available APs are grouped into multiple processing clusters (PCs) to share channel state information (CSI), ensuring that they have knowledge of the CSI for all users assigned to the given cluster for the purposes of designing resource allocation and precoding. We utilize the sum pseudo-SE metric, which accounts for intra-cluster interference and inter-cluster-leakage, providing a close approximation to the true sum   achievable SE. For a given PC, we formulate two optimization problems to maximize the cluster-wise weighted sum pseudo-SE under fronthaul constraints, relying solely on local CSI. These optimization problems  are associated with different computational complexity requirements. The first optimization problem jointly designs precoding, user association, and power allocation, and is performed at the small-scale fading time scale. The second optimization problem optimizes user association and power allocation  at the large-scale fading time scale. Accordingly, we develop a novel application of modified weighted minimum mean square error (WMMSE)-based approach to solve the challenging formulated  non-convex mixed-integer  problems. Numerical results show that (a) the proposed cluster-wise processing solutions  significantly outperform the heuristic approaches under both the statistical and instantaneous CSI-based designs, while statistical CSI-based design is good enough in some network configurations since it provides a better performance/implementation complexity trade-off;  (b) under limited fronthaul, our proposed cluster-wise processing frameworks achieve sum-spectral efficiency (SE)   that is competitive with state-of-the-art  network-wide processing solutions,  while avoiding the high computational complexity of  processing and the heavy overhead of CSI acquisition. 

\let\thefootnote\relax\footnotetext{
The work of  H. Shin was supported by the National Research Foundation of Korea (NRF) grant funded by the Korean government (MSIT) under RS-2025-00556064 and by the MSIT (Ministry of Science and ICT), Korea, under the ITRC (Information Technology Research Center) support program (IITP-2025-RS-2021-II212046
) supervised by the IITP (Institute for Information \& Communications Technology Planning \& Evaluation). (\textit{Corresponding authors: Hyundong Shin; Hien Quoc Ngo}).

Z. Mobini   is  with  the Department of Electrical and Electronic Engineering, The University of Manchester, Manchester M13 9PL, U.K., and also
 with the Centre for Wireless Innovation (CWI), Queen's University Belfast, BT3 9DT Belfast, U.K. (email:zahra.mobini@manchester.ac.uk).

A. gokceoglu1,  L. Wang, and G. Peters are with the  Huawei’s Sweden Research Center, Stockholm, Sweden (e-mail:
\{ahmet.hasim.gokceoglu1, leo.li.wang, gunnar.peters\}@huawei.com).

H.~Shin 
is with the Department of Electronics and Information Convergence Engineering,
Kyung Hee University,
1732 Deogyeong-daero, Giheung-gu,
Yongin-si, Gyeonggi-do 17104, Republic of Korea
(e-mail: hshin@khu.ac.kr). 

H. Q. Ngo   is with the Centre for Wireless Innovation (CWI), Queen's University Belfast, BT3 9DT Belfast, U.K., and is also with the Department of Electronic Engineering, Kyung Hee University, Yongin-si, Gyeonggi-do 17104, Republic of Korea (email: hien.ngo@qub.ac.uk).
}
\end{abstract}

\begin{IEEEkeywords}
  Cell-free massive  multiple-input multiple-output (CF-mMIMO),  cluster-wise processing, fronthaul,  resource allocation.   
\end{IEEEkeywords}
\vspace{-0.5cm}
\section{Introduction}
Cell-free massive multiple-input multiple-output (CF-mMIMO) is a cutting-edge wireless technology developed to address the extensive connectivity requirements and escalating data traffic demands of next-generation wireless networks~\cite{Mohammadali:survey:2024}. In CF-mMIMO, a large number of access points (APs) are distributed over a wide area to simultaneously and coherently serve multiple users. This technology integrates massive MIMO, network MIMO (also known as coordinated multipoint joint transmission, CoMP-JT), and cooperative networks~\cite{hien:2017:wcom}, capitalizing on their strengths to manage interference, provide high macro/micro-diversity gains, and achieve high array gains, ensuring ubiquitous connectivity.

 Since the seminal work on canonical CF-mMIMO~\cite{hien:2017:wcom}, this network topology has become a focal point of extensive research~\cite{Mohammadali:TCOM:2024}. Over the years, a substantial body of literature has emerged, focused on improving the practicality, spectral efficiency (SE), and energy efficiency of CF-mMIMO systems through advanced signal processing techniques and resource allocation strategies~\cite{Ngo:PROC:2024}. Early CF-mMIMO configurations assumed that i) all APs possessed network-wide channel state information (CSI) through a centralized processing unit (CPU) connected by fronthaul links, and ii) all APs transmitted/received information signals to/from all users during the downlink/uplink data transmission phases. While this configuration optimized system performance, it required a high degree of coordination among APs, introduced considerable computational complexity, and necessitated extensive fronthaul/backhaul signaling for CSI and data exchange. These factors posed significant challenges to scalability as the network size increased, whether in terms of the number of APs or users. To address these scalability challenges, a user-centric approach has been proposed in~\cite{Buzzi:WCL:2017,Hien:TGCN:2018, Ammar:TWC:2021}, wherein each user is served by a subset of APs rather than the entire network. This user-centric CF-mMIMO offers a more scalable solution for CF-mMIMO implementation as network sizes grow, while delivering performance that is close to the canonical CF-mMIMO in terms of achievable SE~\cite{Ngo:PROC:2024}.

 Although the theoretical frameworks for CF-mMIMO are well-established, the practical implementation of this technology still faces significant challenges. These challenges include practical and scalable signal processing, resource allocation (power control and user association), CSI knowledge, and fronthaul requirements, particularly as network size increases~\cite{Ngo:PROC:2024,Zahra_TWC_Huawei_2025}. Signal processing and resource allocation in CF-mMIMO systems can generally be implemented using either a centralized approach at the CPU or a distributed approach at each AP.
Centralized processing has been widely utilized in the CF-mMIMO literature for tasks such as precoding/combining design, including centralized minimum mean square error (MMSE) with varying levels of coordination in~\cite{Emil:WCOM:2020} and centralized zero-forcing (ZF) in~\cite{nayebi:2017:wcom,Pei:2020:TWC}, as well as for resource allocation in~\cite{Hien:TGCN:2018, Buzzi:WCOM:2020, Alonzo:2019:TGCN, Mohammadali:JSAC:2023, Hao:IOT:2024, Jiafei:WCL:2025}. While these methods optimize system performance, they also require extensive CSI sharing via fronthaul/backhaul links, which can heavily burden network resources. Moreover, the computational load of performing optimization at the CPU becomes increasingly complex due to the high dimensionality of the aggregated channels. These combined factors make centralized processing both impractical and difficult to scale in real-world applications.

 To address these limitations, distributed processing approaches have been proposed as a more scalable alternative. For example, the authors in~\cite{Atzeni:TWC:2021} introduced a framework for cooperative precoding design in CF-mMIMO systems, which eliminates the need for backhaul signaling for CSI exchange by employing over-the-air (OTA) signaling mechanisms to acquire the necessary information at APs. Furthermore, a distributed precoding design known as Team MMSE was introduced in~\cite{Lorenzo:2022:twc}.  Moreover,   distributed yet simple and sub-optimal power allocation algorithms have been proposed in~\cite{Emil:WCOM:2020,nayebi:2017:wcom}. Scalable fractional power control strategies along with conjugate precoding for downlink CF-mMIMO was proposed in~\cite{Giovanni:2021:SPAWC}, while distributed max-min power control  by training the neural network with only local CSI at each AP was proposed in~\cite{Sucharita:2019:Asilomar}.
While distributed processing offers better scalability and lower complexity by relying only on local channel estimates at each AP, it also has performance limitations compared to centralized methods. This is due to the fact that distributed approaches lack network-wide coordination between APs, which can lead to suboptimal resource allocation and low efficient interference management. The absence of network-wide control means that distributed approaches may not be able to fully utilize the available diversity gains or cancel interference as effectively as centralized approaches. Consequently, system performance may reduce in scenarios where coordination between APs is critical for obtaining optimal performance. Therefore, finding a balance between system performance and the degree of cooperation among APs becomes a key challenge for ensuring the scalability of CF-mMIMO systems.

To this end, partially centralized processing was considered as a promising solution that ensures scalability within the system architecture~\cite{Ngo:PROC:2024}. In particular, the authors in~\cite{Interdonato:TWC:2020} employed a semi-distributed version of partial ZF, which partially mitigates inter-AP interference by sharing a limited amount of CSI over the fronthaul network. However, this comes at the cost of increased computational complexity and fronthaul overhead. Furthermore, the authors of~\cite{Hien:JCOM:2021} harnessed a centralized ZF-based precoding for a subset of APs and distributed maximum ratio transmission precoding  for the remaining APs. However, the
significant drawback of these studies is that they only focused
on either precoding design or power control. Therefore, the study of how to efficiently perform partially centralized processing for more complicated methods such as joint precoding, user association, and power control in a CF-mMIMO system under imperfect fronthaul network is extremely timely and important.



\subsection{Key Contributions}
 To address the need for scalable CF-mMIMO in practical fronthaul-limited scenarios, we are inspired by the novel paradigm of cluster-wise processing~\cite{Ammar:TWC:2022}. Cluster-wise processing involves performing signal processing and resource allocation within a cluster (group) of APs, requiring only partial CSI sharing and potentially reducing the computational requirements due to the smaller dimensionality of the aggregated channels. This approach strikes a balance between having some coordination among APs and using a fully centralized system.  {Cluster-wise processing not only supports scalable implementation but also aids in managing inter-cluster interference by localizing coordination~\cite{Zhang:2024:WCL}}.  One fundamental challenge in designing cluster-wise schemes is the signal-to-interference-plus-noise ratio (SINR) metric, which is inter-dependent across all users and APs in different clusters. Therefore, any optimization framework based on the SINR maximization criterion may not be suitable for cluster-wise resource allocation. An alternative is to use leakage-based metrics~\cite{Sadek:2011:TWC,Ammar:TWC:2022}. Considering the notion of   hybrid signal-to-leakage-and-intra-cluster-interference-and-noise ratio (SLINR), which considers the desired signal, the intra-cluster interference,  and the leakage interfering with other users in other clusters,  the authors in~\cite{Ammar:TWC:2022},  developed     cluster-wise   resource allocation approaches for user-centric cell-free systems for two scenarios: system optimization is carried out at the APs or at multiple central units (CUs) controlling a subset (a cluster) of APs. However, this work did not account for fronthaul constraints and relied solely on instantaneous CSI. While instantaneous CSI-based optimization enables dynamic user association, precoding, and power allocation, it comes with high computational and signaling costs. The need for frequent CSI acquisition increases with the number of antennas, subcarriers, and users.
A more scalable alternative is statistical CSI-based design, which takes advantage of channel hardening to allow efficient resource allocation based on large-scale fading, rather than the fast-varying small-scale fading. This approach significantly reduces computational complexity, as system parameters remain stable over longer time scales, and minimizes the need for real-time downlink CSI estimation, making it particularly suitable for large-scale CF-mMIMO deployments.
However, in certain propagation environments  or for some system setups channel hardening deteriorates~\cite{Zahra:2025:LectureNote}, rendering statistical CSI-based power allocation suboptimal. As a result, instantaneous CSI-based power allocation remains necessary in some CF-mMIMO configurations. Therefore, a comparative analysis of statistical versus instantaneous CSI-based designs remains unexplored for cluster-wise  fronthaul-limited CF-mMIMO.

 In this context, we adopt the hybrid SLINR metric for a fronthaul-limited CF-mMIMO system, where APs are grouped into multiple processing clusters (PCs). We formulate two distinct optimization problems, each operating at a different time scale: 1) Instantaneous CSI-based optimization: Jointly optimizes precoding, user association, and power allocation at the small-scale fading time scale, enabling adaptation to rapid channel variations. 2) Statistical CSI-based optimization: Jointly optimizes power control and user association based on large-scale fading statistics, leveraging the channel hardening property of CF-mMIMO systems. To this end, we introduce a hardening-based pseudo-SE, which simplifies resource allocation while maintaining performance in large-scale networks. For efficient cluster-wise processing, we develop two novel applications of a modified weighted minimum mean square error (WMMSE)-based approach. The key motivation behind using modified WMMSE algorithms is their ability to jointly optimize user association, power allocation, and/or precoding vectors within a unified framework. Furthermore, each iteration of the WMMSE algorithm is computationally efficient, making it feasible for large-scale CF-mMIMO systems without imposing excessive overhead.


 The key contributions of this paper are summarized as
follows:
\begin{itemize}
\item We provide a cluster-based network architecture for a fronthaul-limited   user-centric CF-mMIMO system with multiple-antenna  APs.  The proposed  framework is very general and can cover different CF-mMIMO implementations. In particular, it can be degenerated to different special cases such as  CF-mMIMO with network-wide fully centralized operation  or   fully distributed operation.
\item We formulate two optimization problems   for maximizing the cluster-wise weighted sum pseudo-SE/hardening-based pseudo-SE  under per-AP transmit power  and fronthaul constraints which can be carried out at two different time scales.      Two novel modified WMMSE-based algorithms are then proposed to solve the challenging formulated  non-convex mixed-integer  problems.
\item Numerical results show that the proposed cluster-wise processing solutions significantly outperform the heuristic approaches for both the statistical and instantaneous CSI-based designs. 
They also confirm that CF-mMIMO can be efficiently deployed by utilising our proposed cluster-wise processing  without significant loss in performance but with much lower fronthaul requirements compared to centralized  system. In fact,  they show that cluster-wise  processing  is highly preferable compared to network-wide alternatives, either in the regime of  a high number of APs or high user loads.  In addition, in various simulation setups, e.g., large values of transmit antennas at APs or stringent  fronthaul constraints, the  statistical CSI based design leads to a negligible performance loss, compared to instantaneous CSI-based  design. Therefore, statistical CSI-based design is enough since it provides a better performance/implementation complexity
trade-off. 
\end{itemize}

\textit{Notation:} We use bold upper (lower) case letters to denote matrices (vectors).  The superscripts  $(\cdot)^T$ and $(\cdot)^\dag$ stand for the  transpose and conjugate-transpose, respectively. A zero-mean circular symmetric complex Gaussian distribution having a variance of $\sigma^2$ is denoted by $\mathcal{CN}(0,\sigma^2)$, while $\mathbf{I}_N$ denotes the $N \times  N$ identity matrix.   $\mathbb{E}\{\cdot\}$ denotes the statistical expectation. Finally, operations $\mathrm{Re}(\cdot)$ and $|\cdot|$ represent  the real part  of a complex number and the size (or cardinality) of a set. 
\begin{figure}[t]
	\centering
	\vspace{0em}
	\includegraphics[width=72mm, height=50mm]{  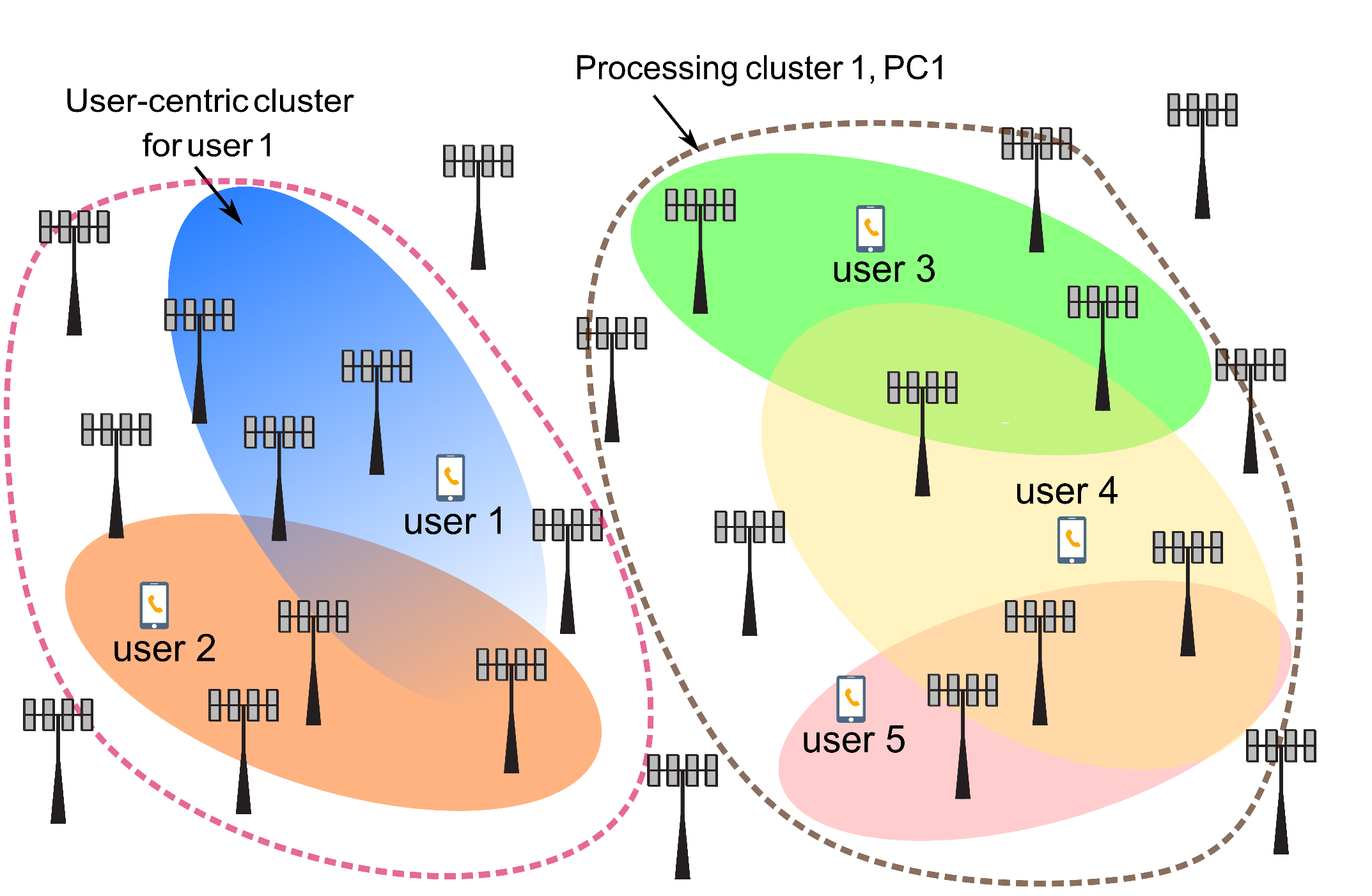}
	\caption{  User centric CF-mMIMO with  cluster-wise processing.}
	\label{fig:Fig1}
\end{figure}
\section{Generic Cluster-Wise Processing-Based Network Architecture }
We consider a  CF-mMIMO system comprising $M$ $L$-antenna APs and $K$ single-antenna users. The sets of APs and users are denoted by $\mathcal{M}=\{1,\cdots,M\}$ and $\mathcal{K}=\{1,\cdots,K\}$, respectively. 
The baseband unit (BBU) functionality is split into two entities: baseband low (BBL) and baseband high (BBH). The BBH handles processing tasks such as precoding, encoding, and radio resource management, while the BBL is responsible for tasks like weight applications, error correction, and modulation. Each BBH is connected to its associated BBL via limited-capacity fronthaul links to transmit information such as precoding vectors, information signals, and power allocation coefficients. In addition, BBHs are interconnected through backhaul links to facilitate information exchange between APs. {It is notable that in this work, we do not impose explicit backhaul capacity constraints to focus on fronthaul limitations, which are usually the main bottleneck in CF-mMIMO systems, especially when fronthaul links are wireless or limited. Backhaul links—typically high-capacity fiber or Ethernet—are assumed sufficient to support intra-cluster signalling~\cite{Park:2014:spm}. Nevertheless, cluster-wise processing for CF-mMIMO systems under joint fronthaul and backhaul constraints is an interesting direction worthy of future research.}

We assume a frequency-flat slow fading channel model for each orthogonal frequency-division multiplexing (OFDM) subcarrier.   Here, we assume that the cyclic prefix  length is greater than the maximum delay spread, ensuring the orthogonality between OFDM subcarriers. For notational simplicity, the subcarrier index will be omitted. Let $\qg_{mk} \in \mathbb{C}^{L \times 1}$ represent the complex channel vector between the $m$-th AP and the $k$-th user. This channel vector can be modeled as
\begin{equation}
 {\qg}_{mk} = \beta ^{1/2}_{mk}\qh_{mk},
\end{equation}
where $\beta_{mk}$ denotes the large-scale fading coefficient that includes path-loss and
shadowing effects, while  $ {\qh}_{mk}\in \mathbb{C}^{L\times 1}$ is the small-scale fading vector whose entries are independent and identically distributed (i.i.d.)  $\mathcal{CN} (0, 1)$ RVs. Large-scale fading coefficients change slowly and may be constant in range of many small-scale fading coherence intervals (over time and frequency bands). Hence, it is assumed that these coefficients are priory known  at each BBL/BBH.

 We consider the concept of cluster-wise processing in conjunction with user-centric association. Specifically, in cluster-wise processing, the available APs are divided into $S$ disjoint PCs represented by the set of sets   $\mathcal{C}=\{\mathcal{C}_1, \mathcal{C}_2, \ldots, \mathcal{C}_S$\}. {APs within each PC $\mathcal{C}_s$  share CSI via intra-cluster backhaul links, ensuring knowledge of the CSI for all users served by APs in that cluster to facilitate resource allocation and precoding\footnote{While this coordination introduces additional overhead, detailed modeling of this overhead is beyond the scope of this paper and is left for protocol-level and cross-layer investigations.}}.
 
 Each user $k$ is associated with only one PC and can be coherently served by the APs, or a subset of APs, within that cluster.
Fig.~\ref{fig:Fig1}  shows a typical cluster-wise processing-based user-centric CF-mMIMO architecture including  five users served by their associated APs  and two PCs. The colored regions indicate which set of APs  serve which users, while the purple and brown regions show the PCs.

\begin {Remark} When $S=1$, there is a single PC consisting of all the APs in the network. This corresponds to the CF-mMIMO system with fully centralized (network-wide) processing. In contrast, when $S=M$, the system operates with fully distributed processing.
\end{Remark}

\begin {Remark}  Unlike traditional cellular networks, where users are typically associated with a single AP, our model allows each user to be coherently served by multiple APs within its  associated PC. This preserves the joint transmission and reception characteristics of CF-mMIMO, enabling macro-diversity and effective interference suppression. While we introduce cluster-wise processing to manage fronthaul limitations and computational complexity, these clusters do not enforce fixed cell boundaries. Instead, they support flexible user-AP associations and maintain the core CF-mMIMO principle of user-centric service. In contrast to clustered CoMP, where coordination is often limited, massive MIMO properties are not applied, and user association is predefined by fixed cell layout, our approach allows dynamic and optimized cooperation among distributed APs within each cluster.
\end{Remark}
\subsection{Uplink Training  for Channel Estimation
}\label{sec:uplinktraining}
In the uplink training phase, all users send pilot signals to the APs. {This full-pilot model is adopted to enable user-centric coordination and coherent joint transmission, in line with foundational CF-mMIMO studies~\cite{Ngo:PROC:2024}. Alternative strategies, such as efficient pilot reuse or partial CSI acquisition, may be employed to reduce overhead while preserving acceptable system performance.} 

Accordingly, each AP can estimate the corresponding channels to all users using the obtained pilot signal. Note that at each AP, the channel estimation is performed at its BBH. We consider orthogonal  pilot assignment\footnote{Orthogonal  pilot assignment assumption  is applicable in many scenarios, particularly when the coherence interval is sufficiently long and/or the number of users is not excessively large. In practice, environments with medium or low mobility often experience a relatively long coherence interval~\cite{marzetta2016fundamentals, Peng:TCOM:2023,Buzzi:WCOM:2020}. However, in high-mobility or dense network settings, non-orthogonal pilots  are necessary, which require more advanced techniques, such as pilot contamination mitigation.}. This requires $\tauu\geq K$, where $\tauu$ is the uplink training duration.   AP $m$ uses the MMSE estimation technique to estimate the channels to all users. The MMSE channel estimate of $\qg_{mk}$, $\hat{\qg}_{mk}$, includes $L$ i.i.d. $\mathcal{CN}(0,\gamma_{mk})$ elements, where $	\gamma_{mk}= 
  \frac{\tauu \Pu \beta_{mk}^{2}}{\tauu \Pu \beta_{mk}+1}$, and $\Pu$ is the  the normalized transmit power of each pilot symbol.


\subsection{Downlink  Payload Data Transmission}
For each user $k$, we choose a PC $\Csm$, where PC selection can be based on diverse criteria, including the value of large-scale fading coefficient $\beta_{mk}$. Our proposed PC  selection strategy will be discussed in section~\ref{Sec:numerical}.
 Also,  we define the set $\Usm \subset \Kset$ representing the  users that need to be served by PC $\Csm$, while $\mathcal{U}_{-s}$ represents the set of users that have no association with APs
in set $\Csm$, i.e., $ \mathcal{U}_{-s} = \mathcal{K}\setminus \Usm$.
It is notable that each AP $m$ in PC $\mathcal{C}_s$ communicates only with a subset of users in  $\Usm$. In other words,    each user $k$ will be served by a subset of APs (not all APs) within  PC $\Csm$, which is referred to as the user-centric cluster for user  $k$. 
 We use the binary variable $a_{mk}$ to show the user assignment for each AP $m$,  so that
\begin{align}
\label{a}
a_{mk} \triangleq
\begin{cases}
  1, & \text{if AP $m$ associates with user $k$,}\\
  0, &  \mbox{Otherwise}, \qquad\forall m, k.
\end{cases} 
\end{align}
%
Therefore, the $L \times 1 $ signal transmitted by the $m$-th AP in $\Csm$\footnote{In what follows, for ease of notation,  we focus on PC $\Csm$ and do not include the PC index in some system parameters.}  can be expressed as
\begin{equation}\label{eq:Sm}
	\qs_{m}=  \sum\nolimits _{k\in \Usm}a_{mk}\sqrt{\eta_{mk}} \qq_{mk}x_{k},  \forall m\in \Csm,
\end{equation}
where $\qq_{mk}$ denotes the downlink precoding vector constructed for the $m$-th AP to the $k$-th user. Moreover, $x_{k}$ is the information symbol  intended for the $k$-th user, $\mathbb{E}\{x_{k}x_{k}^H \} = 1$, and  $\eta_{mk}$  represents  the $k$-th user power control coefficient.
Each AP $m$ is required to meet the power constraint  

\begin{align} \label{eq:Pconst}
   \sum\nolimits_{k \in \Usm} \|a_{mk}\sqrt{\eta_{mk}} \qq_{mk}\|_2^2\leq \Pb,~\ \forall m\in \Csm,
\end{align}
where $\Pb$ is the maximum transmit power  of each AP.
 The  $k$-th user receives signal contributions from  the APs in  $\Csm$; the observable signal is given by
\begin{align}\label{eq:received_r}
r_{k}
=&\underbrace{ 
\sum\nolimits _{m\in \Csm}
a_{mk}\sqrt{\eta_{mk}}
\qg^{H}_{mk} \qq_{mk}x_{k}}_{\text{Desired signal}}+\nonumber\\
&\underbrace{ 
\sum\nolimits _{m\in \Csm}
\sum\nolimits _{\substack {k'\in \Usm \\ k' \neq k}}
a_{mk'}\sqrt{\eta_{mk'}}
\qg^{H}_{mk} \qq_{mk'}x_{k'}}_{\text{Intra-cluster interference}}+\nonumber
\\
&\underbrace{ 
 \sum\nolimits _{m\in \Csnm} \sum\nolimits _{\substack {k'\in\mathcal{U}_{-s} }}\!\!\!
a_{mk'}\sqrt{\eta_{mk'}}
\qg^{H}_{mk} \qq_{mk'}x_{k'}}_{\text{Inter-cluster interference}} +  n_k,
\end{align}
where $\Csnm=\mathcal{C}\setminus \Csm$ and $n_k\sim \mathcal (0, \sigma^2)$ is the  additive white Gaussian noise.   The first term in~\eqref{eq:received_r} represents the desired signal part, while the second term represents the intra-cluster
interference (all the signal components intended
for user $k'\in \Usm, k' \neq k$ from the APs in cluster $\Csm$), and the third term is inter-cluster interference (all the signal components intended
for user $k'\in \Unsm$ from the APs in other clusters than $\Csm$). 
\subsection{Fronthaul Requirements}\label{app:Fronthaul Requirements}
We formulate the total fronthaul requirement at each AP $m$ for transmission from BBH $m$ to BBL $m$  which consists of two parts: 1) 
fronthaul requirement  for downlink data transmission, $\mathrm{FH}_{m,\mathrm{da}}$, and 2) fronthaul requirement for sending the precoding vectors, $\mathrm{FH}_{m,\mathrm{pr}}$. To this end, the  per-AP fronthaul capacity constraint can be written as
\begin{align}\label{eq:Kmax_cons}
&\mathrm{FH}_{m,\mathrm{da}}+\mathrm{FH}_{m,\mathrm{pr}} \leq \mathrm{FH}_{\mathrm{max}},~~~ \forall m \in \Csm,
\end{align}
where $\mathrm{FH}_{\mathrm{max}}$ is the per-AP maximum fronthaul capacity. More specifically, for each AP $m$, the fronthaul consumption for transmitting information symbols to its associated users with $a_{mk}=1$, utilizing packed-based evolved common public radio interface (eCPRI) for the fronthaul  transmission, is given by
   \begin{align}\label{eq:FH1}
\mathrm{FH}_{m,\mathrm{da}}= \frac{\log_2 (M_{\mathrm{mo}}) N_{\mathrm{sub}} N_{\mathrm{o}}\sum_{k \in \Usm} a_{mk}}{\varepsilon_\mathrm{{cp}}\delta _\mathrm{da}},
\end{align}
where $M_{\mathrm{mo}}$ is the modulation cardinality, $N_{\mathrm{sub}}$ is the number of OFDM subcarriers, {while  $N_{\mathrm{o}}$ is the number of OFDM symbols}. Moreover, $\delta_{\mathrm{da}}$  shows the  transmit
delay for the data,   and $\varepsilon_{\mathrm{cp}}$ is the efficiency of the  CPRI. 
In addition, the required fronthaul requirement for sending the precoding vectors can be written as
\begin{align}\label{eq:FH2}
\mathrm{FH}_{m,\mathrm{pr}}= \frac{2L \sum_{k\in \Usm} a_{mk} N_{\mathrm{bits}}N_{\mathrm{Gran}}}{\varepsilon_\mathrm{{cp}}\delta_\mathrm{{pr}}} ,
\end{align}
where $N_{\mathrm{bits}}$ denotes the number of quantization bits, while $N_{\mathrm{Gran}}$ shows the precoding granularity, $\delta_\mathrm{{pr}}$  shows the transmit delay of the precoding weights.
One intuitive observation from~\eqref{eq:FH1} and~\eqref{eq:FH2}  is that the fronthaul consumption scales with the total number of users that each AP serves. Accordingly, in a fronthaul limited CF-mMIMO system with the  fronthaul constraint~\eqref{eq:Kmax_cons}, we must constrain the downlink traffic bandwidth to be fronthauled  by restricting the number of users each AP  $m \in \Csm$  serves  as
\begin{align}\label{eq:Kmax}
\sum\nolimits_{k \in \Usm} a_{mk}  \leq \left\lfloor\frac{\mathrm{FH}_{\mathrm{max}}}{ \left(\alpha_1 \log_2 (M_{\mathrm{mo}}) + \alpha_2\right)}\right\rfloor \triangleq \Kmax,
\end{align}
where $\lfloor \cdot \rfloor$ is the floor function, $\alpha_1\triangleq  \frac{N_{\mathrm{sub}} N_{\mathrm{o}}}{\varepsilon_\mathrm{{cp}}\delta_\mathrm{{da}}}  $, and $\alpha_2 \triangleq \frac{2 L N_{\mathrm{bits}}N_{\mathrm{Gran}} }{\varepsilon_\mathrm{{cp}}\delta_\mathrm{{pr}}} $.

On the other hand, a higher $M_{\mathrm{mo}}$ is essential to achieve the high SE in the system. However, from~\eqref{eq:FH1}, fronthaul consumption for sending information symbols increases with $M_{\mathrm{mo}}$, thus, there is a trade-off  between the SE and fronthaul requirement.
From the information-theoretic perspective,  modulation-constrained achievable SE, $\mathcal{R}_{\mathrm{mo}}$, is limited by  the AWGN channel capacity, $\bar{\mathcal{C}}$, such that $\mathcal{R}_{\mathrm{mo}} < \bar{\mathcal{C}}$.
Moreover, $\mathcal{R}_{\mathrm{mo}}$ cannot exceed the entropy of the modulation constellation, i.e., $\mathcal{R}_{\mathrm{mo}}< \log_2 (M_{\mathrm{mo}})$. Accordingly,
$\mathcal{R}_{\mathrm{mo}}$ can be upper-bounded as
$\mathcal{R}_{\mathrm{mo}} \leq \mathrm{min} (\bar{\mathcal{C}},\log_2 (M_{\mathrm{mo}}))$~\cite{Urlea:2021:JLT}.
{Therefore, we consider the following   achievable SE constraint 
\begin{align}\label{eq:FHMorder}
R_{k} \leq \log_2 (M_{\mathrm{mo}}),
\end{align}
where $R_{k}$ denotes the achievable SE for user $k$, in addition to the fronthaul constraint given in~\eqref{eq:Kmax}.}
\section{Instantaneous CSI-Based Cluster-Wise   Sum pseudo-SE Maximization }
 Here, we assume that the users perfectly know the  channels. These channels  can be estimated by the users using downlink pilots.   We would like to highlight that we also assume   the channel estimation error at both users and APs is very small, i.e., the estimated channel $\hat{\qg}_{mk}$   is almost identical to the true channel    $\qg_{mk}$, $\forall m, k$. This implies that the current results for the instantaneous CSI-based design serve as an upper bound.
Given the instantaneous CSI, the    achievable SE for user $k$ can be written as:
\begin{align}\label{eq:perfect:SE}
R_k=\log_2 (1+{\SINRk}),
\end{align}
 where
\begin{align}
~\label{eq:Perfect:SINRk}
{\SINRk}\!\!=\!\frac{ \big|\sum_{m \in \Mset}a_{mk}\sqrt{\eta_{mk}} \qg_{mk}^{H} \qq_{mk}\big|^{2}}
{ \sum\nolimits _{\substack{k'\in \Kset \\ k'\neq k}}\!\!\big|\!\sum_{m \in \mathcal{M}}a_{mk'}\sqrt{\eta_{mk'}}  \qg_{mk}^{H} \qq_{mk'}\!\big|^{2}\!\!+\!\Sn}.
\end{align}
\subsection {Hybrid Leakage-Intra-Cluster-Interference}
Based on~\eqref{eq:received_r} the strength of the desired signal component for user $k \in \Usm$ is given by
\begin{align}
\mathrm{DS}_{k}=&  \left|\sum\nolimits_{m \in \Csm}a_{mk}\sqrt{\eta_{mk}} \qg_{mk}^{H} \qq_{mk}\right|^{2}.
\end{align}
while the strength of the intra-cluster interference at  user $k$  can be written as
\begin{align}
\mathrm{ICI}_{k}= \sum\nolimits_{\substack {k'\in \Usm \\ k' \neq k}}\left|\sum\nolimits_{m \in \Csm}a_{mk'}\sqrt{\eta_{mk'}}  \qg_{mk}^{H} \qq_{mk'}\right|^{2}.
\end{align}
Now, we define  the quantity, called  \emph{leakage} interference experienced by the users in $\mathcal{U}_{-s}$ from APs in $\Csm$ by serving user $k$ as
\begin{align}
\mathrm{L}_{k}=\sum\nolimits_{k'\in  \Unsm}\left|\sum\nolimits_{m \in \Csm}\sqrt{t_{k'}}a_{mk}\sqrt{\eta_{mk}}  \qg_{mk'}^{H} \qq_{mk}\right|^{2},
\end{align}
where the binary parameter $t_{k'}$ represents the association assumption about user $k' \in \Unsm$. In particular, $t_{k'}=1$ if user $k'$ is scheduled by at least one AP in $\mathcal{C}_{-s}$;  $t_{k'}=0$ otherwise. We thus define hybrid expressions 
in terms of so-called \emph{SLINR} that account for both the leakage interference and intra-cluster interference as
\begin{align}\label{eq:SLINRk}
\SLINRk=\frac{ \mathrm{DS}_{k}}
{\mathrm{ICI}_{k}+\mathrm{L}_{k}+\Sn}.
\end{align}
The  SLINR expression  in~\eqref{eq:SLINRk} depends only on locally constructed precoding vector and local CSI in each PC.
Accordingly, a  pseudo-SE  between each  PC  $\Csm$ and its user $k \in \Usm$ is defined as
\begin{align}\label{eq:SLINR}
\zeta_{k}=\log _{2} (1+\SLINRk).
\end{align}

\subsection{Problem Formulation}
In this section, we aim to jointly  optimize precoding, downlink transmit powers, and user association   for the maximization of the system sum pseudo-SE, subject to per-AP fronthaul capacity and maximum
transmit power constraints for the given PC $\mathcal{C}_s$. Accordingly, in what follows we formulate joint optimization problem for a given small-scale fading coherence time. 
We would like to highlight that the sum pseudo-SE optimization is motivated by two main factors. First, it enables the development of a cluster-wise processing by relying solely on local variables (i.e., precoding vectors, power allocation coefficients, and CSI) in the given PC. In fact, in actual   SE-based designs, the APs should have global CSI and precoding knowledge. In practice, it is not easy/scalable to measure/obtain the CSI between different APs and users at each AP, especially when there are high number of users or APs in the network. Second, pseudo-SE criterion emphasizes the importance of maximizing the useful signal while minimizing leakage and intra-cluster interference.   By minimizing leaked interference to other users, users in other PCs  improve  their SINR by reducing the interference they experience. Additionally, addressing intra-cluster interference benefits users within the same PC. Balancing leakage and intra-cluster interference prevents uniform scaling of the AP's beam power in both signal and leakage terms during optimization~\cite{Ammar:TWC:2022}.

 We note that, for calculating the pseudo-SE of user $k \in \Csm$ as given in~\eqref{eq:SLINR},  PC $\Csm$ requires the parameter $t_{k'}$, which accounts for the association decisions of users in other PCs, i.e., $\mathcal{C}_{-s}$. In centralized resource allocation, each AP can have knowledge about the association of all the users. However, in cluster-wise processing, a PC might not know the user associations in other PCs. To this end, for the cluster-wise resource allocation, it is reasonable to assume that $t_{k'}=1$, i.e., user $k'$ is scheduled by at least one of the APs in its serving cluster.



For convenience, let  $\aset = \{a_{mk}: m \in \Csm, k \in \Usm\}$  denote the  user-association control variable for PC $\mathcal{C}_s$, $\boldsymbol {\eta }_s$ denote the set of power control coefficients, $\boldsymbol {\eta}_s=\{\eta_{mk}: m \in \Csm, k \in \Usm\}$. Also, let $\Qset_s = \{\qq_{mk}: m \in \Csm, k \in \Usm\}$  denote  the collective precoding vector from  APs in PC $\Csm$ to user $k$. Accordingly, for the  PC $\Csm$, the joint  optimization problem  can be formulated as
\begin{subequations}\label{eq:problem2}
\begin{align}
\max _{\boldsymbol {\eta}_s, \aset, \Qset_s} ~&\sum\nolimits_{k \in \Usm}    w_k\zeta_{k}(\boldsymbol {\eta}_s, \aset, \Qset_s) \label{eq:problem2:O1}
\\
&\hspace{-2em}\text {st.}\hspace{2em}\eta _{mk}\geq 0, \qquad  m\in \mathcal{C}_s,~k\in\mathcal{U}_s, \label{eq:problem2:C1} 
\\
&\hphantom {\text {st.}}R_{k} \leq \log_2 (M_{\mathrm{mo}}), \qquad k\in\mathcal{U}_s,  \label{eq:problem2:C2} 
\\
&\hphantom {\text {st.}}  \sum\nolimits_{k \in \Usm}a_{mk}\leq \Kmax, \qquad \forall m\in \mathcal{C}_s, \label{eq:problem2:C3} 
\\
&\hphantom {\text {st.}} \sum\nolimits _{k\in \Usm} \left \|a_{mk}\sqrt{\eta_{mk}}\qq_{mk}\right \|_{2}^{2} \leq \Pb, \qquad  \forall m\in \mathcal{C}_s. \label{eq:problem2:C4}
\end{align}
\end{subequations}
where $w_k$ presents the priority weight associated with user $k$.

{Fronthaul constraint~\eqref{eq:problem2:C2} includes the actual  achievable SE $R_{k},\ \forall k\in\mathcal{U}_s$, which inherently couples all users and APs across different  PCs. This strong interdependence makes it challenging to incorporate the constraint directly within a decentralized or cluster-wise optimization framework. As such, solving the optimization with this constraint in place would render the proposed scalable, distributed algorithm design infeasible.  To address this issue, we investigated three strategies: 
1)  Relaxation approach: we solve the cluster-wise optimization without enforcing constraint~\eqref{eq:problem2:C2} during the iterative algorithm, and enforce it in a post-processing step on the resulting power allocation; 
2)  Approximation approach:  we approximate $R_k$ by an upper bound $R_k^{\mathrm{up}}$, which considers only intra-cluster interference and ignores inter-cluster effects, thus making it a local function. This allows replacing the actual constraint with $R_k^{\mathrm{up}} \leq \Kmax$; and
3)  Replacement approach:  we replace $R_k$ with the pseudo-SE $\zeta_k$, which captures both intra-cluster interference and inter-cluster leakage.
Our simulation results show that the approximation and replacement approaches lead to   performance degradation—approximately $23\%$ and $10\%$ reductions, respectively, in the system’s sum SE (evaluated for $L=20$, $K=8$, $M=10$, $\mathrm{FH_{max}}=10$ Gbps, and $M_{\mathrm{mo}}=32$), compared to the relaxation approach.   Importantly, the $R_k$ obtained from the optimization still represents the \emph{achievable SE}. When the post-processing step caps $R_k$ to $\log_2(M_{\mathrm{mo}})$, the system transmits at a SE lower than the achievable SE, which provides an additional safety margin against channel impairments and estimation errors, thereby  increasing the transmission reliability. Therefore, although it simplifies the methodology, we adopt the relaxation approach due to its better performance and ease of deployment. Notably, our post-processing step ensures that the final solution satisfies the modulation order and fronthaul constraints, ensuring feasibility in practical implementations. In  particular, we  consider the  optimization problem}

\begin{subequations}\label{eq:problem3}
\begin{align}
\max _{\boldsymbol {\eta}_s, \aset, \Qset_s} &\sum\nolimits_{k \in \Usm} \omega_k  \zeta_{k}(\boldsymbol {\eta }_s,\aset,\Qset_s) \label{eq:problem3:O1}
\\
&\hspace{-2em}\text {st.}\qquad\eta _{mk}\geq 0,\qquad  \forall m\in \mathcal{C}_s,~k\in\mathcal{U}_s, \label{eq:problem3:C1} 
\\
&\hphantom {\text {st.}} \sum\nolimits_{k \in \Usm}a_{mk}\leq \Kmax, \quad \forall m\in \mathcal{C}_s, \label{eq:problem3:C3} 
\\
&\hphantom {\text {st.}}  \sum\nolimits _{k\in \Usm} \left \|a_{mk}\sqrt{\eta_{mk}}\qq_{mk}\right \|_{2}^{2} \leq \Pb, ~\quad  \forall m\in \mathcal{C}_s, \label{eq:problem3:C4}
\end{align}
\end{subequations}
and    then impose the fronthaul constraint as a post-processing step,  given by  
\vspace{-0.5em}
\begin{align} 
    R_{k, \mathrm{post}} = \min({R_{k},\log_2 (M_{\mathrm{mo}})}).
\end{align}

Optimization problem~\eqref{eq:problem3} is  obviously mixed-integer and non-convex due to the user association control variables $a_{mk},~m \in \Csm, k \in \Usm$,   the presence of unknown optimization
variables $\qq_{mk}$ and $\eta_{mk}, m \in \Csm, k \in \Usm$, appearing as products   in  both nominator and denominator of $\zeta_{k}$, and the non-convex constraints \eqref{eq:problem3:C3} and  \eqref{eq:problem3:C4}.
To this end, we first define $\qbar_{mk}\triangleq  a_{mk}\sqrt{\eta_{mk}}\qq_{mk}$.  Hence, constraint \eqref{eq:problem3:C4}  is transformed to a convex one as
\begin{equation}\label{const_powORU_convex}  
\sum\nolimits_{k \in \Usm}{\|\qbar_{mk}\|_2^2}\leq \Pb.
\end{equation}
Moreover, $\zeta_{k}(\boldsymbol {\eta}_s, \aset, \Qset_s)$ in the objective function of~\eqref{eq:problem3} can be re-expressed in terms of $\qbar_s\triangleq\{\qbar_{mk}: m\in \Csm,  k \in \Usm\}$  as~\eqref{eq:SLINRk2} on top of the next page.

\begin{figure*}
\begin{align}
~\label{eq:SLINRk2}
\zeta_{k}(\qbar_{s})=\log_2\Bigg(1+\frac{ \big|\sum_{m \in \Csm}  \qg_{mk}^{H} \qbar_{mk}\big|^{2}}
{ \sum_{\substack {k'\in \Usm \\ k' \neq k}}\big|\sum_{m \in \Csm}   \qg_{mk}^{H} \qbar_{mk'}\big|^{2}+ \sum\nolimits_{k'\in  \Unsm}\big|\sum\nolimits_{m \in \Csm}\sqrt{t_{k'}}   \qg_{mk'}^{H} \qbar_{mk}\big|^{2}+\Sn}\Bigg). \tag{22}
\end{align}
	\hrulefill
    \vspace{-1em}
\end{figure*}

To deal with constraint~\eqref{eq:problem3:C3}, we notice that  each user $k$ is served by AP $m$, $m \in \Csm$, if and only if its precoding vector $\qq_{mk}$ (or equivalently its associated $\qbar_{mk}$) is nonzero. In other words, we can characterize the user  association  by the
indicator function

\setcounter{equation}{22}
\vspace{-0.2em}
\begin{equation}
 {1\kern -0.27em 1}\left \{ \left \|\qbar_{mk}\right \|_{2}^{2} \right \} = \left \{ \begin{array}{l} 0, \quad \text {if} ~ \left \|\qbar_{mk}\right \|_{2}^{2} = 0 \\ 1, \quad \text {otherwise}. \end{array} \right.
 \end{equation}
Accordingly, the  per-AP fronthaul constraint in~\eqref{eq:problem3:C3} can be casted as
\vspace{-0.2em}
\begin{equation}
 \sum\nolimits _{k \in \Usm} {1\kern -0.27em 1}\left \{ \left \|\qbar_{mk}\right \|_{2}^{2} \right \}  \leq \Kmax, \quad \forall m \in \Csm.
 \end{equation}
 In this way, the problem of determining user-AP association $\aset$,   precoding vectors  $\Qset_s$, and power control coefficients  $\boldsymbol {\eta}_s$   is integrated into a single task\footnote{Incorporating user association within the precoding optimization framework does not increase the overall complexity of the algorithm, as it will be discussed in Subsection V-C.} determining the precoding vector $\qbar_{mk}$, $m \in \Csm, k \in \Usm$, for each user $k$. Accordingly, Problem~\eqref{eq:problem3} can be equivalently reformulated as
\vspace{-0.5em}
\begin{subequations}\label{eq:problem4}
\begin{align}
\mathop {\mathrm {max\kern 0pt}} _{\left \{\qbar_{mk} | m\in \Csm, k\in \Usm\right \}} & \sum\nolimits _{k\in\Usm} w _{k} \zeta_{k} (\qbar_s) \label{eq:problem4:C1}
\\
 &\hspace{-2.5em} \mathop {\mathrm {st.\kern 0pt}}\quad  \!\!~\sum\nolimits _{k\in \Usm} \left \|\qbar_{mk}\right \|_{2}^{2} \leq \Pb, ~~\forall m\in \Csm, \label{eq:problem4:C2}
\\
&\sum\nolimits _{k\in \Usm} {1\kern -0.27em 1}\left \{ \left \|\qbar_{mk}\right \|_{2}^{2} \right \} \leq \Kmax,~\forall m\in \Csm.  \label{eq:problem4:C3}
\end{align}
\end{subequations}


\vspace{-2em}
\subsection{Modified WMMSE-Based Approach} \label{sec:CLWMMSE} 
Problem~\eqref{eq:problem4} is still difficult to solve due to the non-convex objective function and fronthaul constraint~\eqref{eq:problem4:C3}. To deal with this issue, we first approximate~\eqref{eq:problem4:C3} and then reformulate the objective function into an equivalent form by exploiting the
 WMMSE  criterion.  In particular, we equivalently expressed  the indicator function in  discrete constraint $\sum _{k\in \Usm} {1\kern -0.27em 1}\left \{ \left \|\qbar_{mk}\right \|_{2}^{2} \right \} \leq \Kmax$ as an
$\ell_0$-norm of a scalar as
\begin{align}
{1\kern -0.27em 1}\left \{ \left \|\qbar_{mk}\right \|_{2}^{2} \right \} = \left \| \left \|\qbar_{mk}\right \|_{2}^{2} \right \|_{0}, 
\end{align}
where $\ell_0$-norm is the number of nonzero elements in a vector. Then, we use the re-weighted $\ell_1$-norm approximation technique as $\left \| \mathbf {x} \right \|_{0} \approx \sum _{i} \vartheta _{i} |x_{i}|$, where $x_i$ denotes the $i$-th element of vector $x$ and $\vartheta_i$ is the weight associated with $x_i$, to approximate a nonconvex $\ell_0$-norm   by a convex  $\ell_1$-norm.   Therefore, the fronthaul constraint~\eqref{eq:problem4:C3} is reformulated as
\begin{equation}\label{eq:lz}
 \sum\nolimits _{k\in \Usm} \vartheta_{mk} \left \|\qbar_{mk}\right \|_{2}^{2}  \leq \Kmax,
 \end{equation}
where $\vartheta _{mk}$ is a constant weight associated with the $m$-th AP  and the $k$-th user and is updated iteratively based on
\vspace{-0.1em}
\begin{equation}\label{eq:WMMSE_vark}
\vartheta_{mk} = \frac {1}{\left \|\qbar_{mk}\right \|_{2}^{2} + \epsilon },
 \end{equation}
with $\epsilon$ is a small constant regularization factor.  It prevents a zero-valued component in $\left \|\qbar_{mk}\right \|_{2}^{2}$ from strictly blocking the nonzero estimate in subsequent iterations. It is notable that the performance of $\ell_1$-norm approximation in~\eqref{eq:lz} is not highly sensitive to the value of $\epsilon$~\cite{Binbin:Access:2014}.  In addition, 
the  weight update rule~\eqref{eq:WMMSE_vark} is based on the fact that setting $\vartheta_{mk}$ to be inversely related to the transmit power level $\|\qbar_{mk}  \|_{2}^{2}$
ensures that APs with lower power transmission to user 
 $k$ are assigned higher weights. As a result, these APs are pushed to further decrease their transmit power to   user $k$ over consecutive iterations.
 

\emph{1) Problem Transformation:}
Now, we can reformulate the optimization problem
\begin{subequations}\label{eq:problem5}
\begin{align}
\mathop {\mathrm {max\kern 0pt}} _{\left \{\qbar_{mk} | m\in \Csm, k\in \Usm\right \}} & \sum\nolimits _{k\in\Usm} w _{k} \zeta_{k} (\qbar) \label{eq:problem5:C1}
\\
 &\hspace{-3.5em} \mathop {\mathrm {st.\kern 0pt}}\hspace{2em}  \!\!~\sum\nolimits _{k\in \Usm} \left \|\qbar_{mk}\right \|_{2}^{2} \leq \Pb, ~~\forall m\in \Csm, \label{eq:problem5:C2}
\\
& \sum\nolimits _{k\in \Usm} \vartheta_{mk} \left \|\qbar_{mk}\right \|_{2}^{2}  \leq \Kmax,~\forall m\in \Csm, \label{eq:problem5:C3}
\end{align}
\end{subequations}
as an equivalent WMMSE problem and use the block coordinate descent method to reach
a stationary point of~\eqref{eq:problem5}. The equivalence between weighted sum SE maximization and WMMSE for MIMO interfering channels is established in~\cite{Razaviyayn:ITSP:2013}. By adopting a similar methodology to that in~\cite{Razaviyayn:ITSP:2013}, it can be readily observed that the generalized WMMSE equivalence presented in~\cite{Razaviyayn:ITSP:2013} also applies to the problem defined in~\eqref{eq:problem5}, which incorporates the  introduced sum-pseudo-SE  maximization objective and the weighted per-AP power constraint~\eqref{eq:problem5:C3}. More precisely, the traditional WMMSE approach considers the mean square  error (MSE) at each user $k$ as
\vspace{-0.2em}
\begin{align}\label{eq:e_MSE}
\hat{e}_{k} =\Ex\{|\hat{u}_k r_k-x_k|^2\},
\end{align}
where $\hat{u}_k$ is the receiver
weight, and takes the sum of the errors over all the users to get the
final cost function. Now, we take another view point and consider  the modified mean MSE at each user $k$ as $e_{k} =\Ex\{|u_k \bar{r}_k-x_k|^2\}$   where $\bar{r}_k$ is the pseudo-received signal at user $k$, which includes the desired signal, intra-cluster interference, and the leakage interference experienced by the users in $\Unsm$ from APs in $\Csm$
by serving user $k$ as
\vspace{-0.2em}
\begin{align}
    \bar{r}_k=&\underbrace{\sum\nolimits _{m\in \Csm} a_{mk}\sqrt{\eta_{mk}} \qg^{H}_{mk} \qq_{mk}x_{k}}_{\text{Desired signal}}+\nonumber\\
    &\underbrace{\sum\nolimits _{\substack{k'\in \Usm , k' \neq k}} \sum\nolimits _{m\in \Csm}a_{mk'}\sqrt{\eta_{mk'}}\qg^{H}_{mk} {\qq}_{mk'} x_{k'}}_{\text{Intra-cluster interference}}+\nonumber\\
    &\underbrace{\sum\nolimits _{k'\in \mathcal{U}_{-s}} \sum\nolimits _{m\in \Csm}t_{k'}a_{mk}\sqrt{\eta_{mk}}\qg^{H}_{mk'} {\qq}_{mk}\bar{x}_{k'}}_{\text{Leakage interference}},
\end{align}
where $\bar{x}_{k'}$ is the pseudo-information symbol for user $k'$.
The motivation to look for an alternative MSE in~\eqref{eq:e_MSE} is to obtain an equivalent WMMSE-based approach for the maximization of sum-pseudo SE in~\eqref{eq:problem5}. 
The equivalence  is explicitly stated in Proposition~\ref{Prop:WMMSE}. Before  proceeding, let us   introduce the  notations $\qbar_{s,k}$ as a collective cluster-wide $\qbar_{mk}$ vector    from  APs in $\Csm$ to user $k \in \Usm$,  $\qbar_{s,k}\triangleq [\qbar_{mk}: m\in \Csm]$,  and $\qg_{s,k}$   as a collective channel vector    from  APs in $\Csm$ to user $k \in  \Kset$.
\begin{proposition}\label{Prop:WMMSE}
The weighted sum-pseudo-SE maximization problem~\eqref{eq:problem5} has the same  solution as the following WMMSE problem:
\begin{subequations}
  \begin{align} \label{eq:WMMSE1}
\mathop {\mathrm {min\kern 0pt}} _{\left \{\rho _{k}, u_{k}, \qbar_{s,k}|  k\in \Usm\right \}} & \sum\nolimits _{k\in \Usm} w _{k} \left (\rho _{k} e_{k} - \log \rho _{k} \right )
\\ 
&\hspace{-3.7em}\mathop {\mathrm {st.\kern 0pt}} \hspace{2.5em}\!\! \sum\nolimits _{k\in\Usm} \left \|\qbar_{mk}\right \|_{2}^{2} \leq P, ~\quad \forall m \in \Csm, 
\\ 
& \sum\nolimits _{k\in \Usm} \vartheta _{mk} \left \|\qbar_{mk}\right \|_{2}^{2} \leq \Kmax,\quad \forall m \in \Csm,
\end{align}   
\end{subequations}
where $\rho_{k}$ represents the MSE weight for user $k$  and   $e_{k}$  shows the corresponding MSE,  which is given by
\begin{align} \label{eq:e}
e_{k} &=\Ex\{ |u_k\bar{r}_k-x_k |^2\}\nonumber\\
 &= {u}_{k}^{2} \Big ( \sum _{j\in \Usm} \qg_{s,k} \qbar_{s,j} \qbar_{s,j}^{H} \qg_{s,k}^{H} 
+
\sum _{j\in \Unsm} \qg_{s,j} \qbar_{s,k} \qbar_{s,k}^{H} \qg_{s,j}^{H}
 \nonumber\\
&~~+\sigma ^{2}  \Big ) 
- 2 \mathrm{Re} \left \{  {u}_{k} \qg_{s,k} \qbar_{s,k}\right \} + 1. 
\end{align}
\end{proposition}
\begin{proof}
    The proofs can be obtained by following similar steps as in~\cite{Razaviyayn:ITSP:2013} for the maximization of weighted sum SE and thus is omitted.
\end{proof}
We highlight that a key advantage of reformulating the sum-pseudo-SE maximization problem~\eqref{eq:problem5} as the equivalent WMMSE problem~\eqref{eq:WMMSE1} is that~\eqref{eq:WMMSE1} exhibits convexity with respect to each individual optimization variable $\rho_k$, $u_k$, and $\qbar_{mk}$ while holding others fixed. This convexity  facilitates the efficient solution of~\eqref{eq:WMMSE1} through the block coordinate descent method, whereby the optimization over $\rho_k$, $u_k$, and $\qbar_{mk}$ is conducted iteratively, as it be outlined as follows:
\begin{itemize}
\item For fixed values of $\qbar_{s,k}$ and $\rho_k$, the objective function in~\eqref{eq:WMMSE1} can be minimized with respect to $u_k$ by setting its first-order derivative to zero, resulting in the MMSE receiver
\begin{align} \label{eq:u}
u_{k}^\opt  &=  \Big (\!\sum _{j\in \Usm}\! \qg_{s,k} \qbar_{s,j} \qbar_{s,j}^{H} \qg_{s,k}^{H} \!+\!\!\!
\sum _{j\in \Unsm} \!\qg_{s,j} \qbar_{s,k} \qbar_{s,k}^{H} \qg_{s,j}^{H} \nonumber\\
&~ ~ ~+ \sigma ^{2}\Big )^{-1} \qg_{s,k} \qbar_{s,k}, ~ \forall k \in \Usm.\quad 
\end{align}

\item  The optimum value for $\rho_{k}$ that minimizes the objective function of~\eqref{eq:WMMSE1} for the fixed $u_k$ and $\qbar_{s,k}$ is obtained by taking the first-order derivative and equating it to zero as
 \vspace{-0.5em}
    \begin{equation}\label{eq:rho}
\rho^\opt_{k} = e_{k}^{-1}, \qquad \forall k \in \Usm.
\end{equation}

\item  As we discuss earlier, the WMMSE method for maximizing the sum-pseudo SE involves iteratively updating one of the three sets of variables $\rho_k$, $u_k$, and $\qbar_{mk}$—while holding the others constant, to approach a local optimum. The update for $u_k$ has been detailed in~\eqref{eq:u}, and the update for $\rho_k$ in~\eqref{eq:rho}. The next step is to determine the optimal values for $\qbar_{mk}$, given the current values of $u_k$ and $\rho_k$. To do this, we substitute the expression for $e_k$ from~\eqref{eq:e} into the objective function in~\eqref{eq:WMMSE1}. The resulting optimization problem  for finding the optimal transmit beamformer $\bar{\qq}_{mk}$ 
is a quadratically constrained quadratic programming (QCQP) problem as 
 \begin{subequations}
    \begin{align}\label{eq:WMMSE2}
  \mathop {\mathrm {min\kern 0pt}} _{ \{\qbar_{mk}| m\in \Csm, k\in \Usm \}} &\sum _{k\in \Usm}\!\! \qbar_{s,k}^{H} \Big ( \sum _{j\in \Usm} w _{j} \rho _{j} {u}_{j}^2\qg_{s,j}^{H}   \qg_{s,j}
 \!+\!w _{k} \rho _{k}\nonumber\\
&\hspace{-6.5em}\times \sum _{j\in \Unsm} \!\!\!u_{k}^2 \qg_{s,j}^{H}   \qg_{s,j}\Big ) \qbar_{s,k}
 - \!\!2 \sum _{k\in \Usm} \!\!w _{k} \rho _{k} \mathrm{Re} \left \{ u_{k} \qg_{s,k}^H \qbar_{s,k}\right \} 
 \\
 &\hspace{-3.5em}{\mathrm {st.\kern 0pt}}\hspace{0.7em}\sum\nolimits _{k\in \Usm} \left \|\bar{\qq}_{mk}\right \|_{2}^{2} \leq P, ~~\forall m\in \Csm 
 \\
&\hspace{-1.5em}\sum\nolimits _{k\in \Usm} \vartheta _{mk} \left \|\bar{\qq}_{mk}\right \|_{2}^{2} \leq \Kmax, \forall m \in \Csm,
\end{align}  
 \end{subequations}
which can be addressed using a standard convex optimization tool like CVX.
\end{itemize}
The solution of~\eqref{eq:WMMSE1}  is
summarized in  \textbf{Algorithm \ref{Alg:WMMSE1}}.\footnote{While this work employs a WMMSE-based optimization approach, using alternative optimization methods such as fractional programming  or  alternating direction method of multipliers (ADMM) approach for solving the QCQP subproblem  could be explored in future research.}

\begin{algorithm}[t]\label{Alg:WMMSE1}
\caption{Instantaneous CSI-Based Sum-Pseudo-SE Maximization With Modified WMMSE Approach at PC $\Csm$}\label{Alg:WMMSE1}
\begin{algorithmic}[1]
\State \textbf{Initialize}: $\vartheta_{mk}^{(0)}$, $\qbar_{mk}^{(0)}, \forall k \in \mathcal{U}_s, \forall m \in \mathcal{C}_s$, iteration index $i = 0$, convergency
accuracy $\xi$.
\While  {$\frac{\big|\sum_{k \in \mathcal{U}_s}w_k \zeta_k^{(i)}-\sum_{k \in \mathcal{U}_s} w_k \zeta_k^{(i-1)}\big|} {\sum_{k \in \mathcal{U}_s} w_k \zeta_k^{(i-1)}}<\xi$}
\State $i = i + 1$;
\State Calculate $u_k^{(i)}$  according to \eqref{eq:u} with $\qbar_{mk}$ fixed, $k\in \Usm$
\State Calculate $e_k^{(i)}$ according to \eqref{eq:e} with $\qbar_{mk}$ and $u_k$ fixed, $k\in \Usm$. 
\State  Update $\rho_k$ according to \eqref{eq:rho}.
\State Calculate the optimal transmit beamformer $\qbar_{mk}$ with  $\qu_k$ and $\rho_k$ fixed $\forall k\in \Usm, m \in \Csm$, by solving the problem~\eqref{eq:WMMSE2}.
\State Calculate $\vartheta_{mk}^{(i)}$ according to \eqref{eq:WMMSE_vark}.
 \EndWhile
\end{algorithmic}
\end{algorithm}

\section{Statistical CSI-Based Cluster-Wise Sum Hardening-Based Pseudo-SE Maximization}
In this section, each AP $m$ uses the channel estimates $\hat{\qg}_{mk}$  in Section~\ref{sec:uplinktraining} to precode the information signals before transmitting them to its assigned users.
Moreover, each  user $k$ relies only on statistical CSI to detect $x_k$ from the received signal in~\eqref{eq:received_r}.  This eliminates the need for users to know the instantaneous channel estimates, which in turn reduces the amount of information that must be exchanged.  
Accordingly, by utilizing the widely adopted hardening bounding technique~\cite{hien:2017:wcom}, the achievable SE for user $k$ can be expressed in closed form as~\cite{  Emil2021JCS}:
\begin{align}\label{eq:perfect:SE_Stat}
\tilde{R}_k=\log_2 (1+\widetilde{\SINRk}),
\end{align}
where
\vspace{-0.5em}
\begin{equation}\label{eq_emil}
    \begin{aligned}
        \widetilde{\SINRk} =\frac{(\qd_k^T\qp_k)^2}{\sum\nolimits_{j \in \Kset} \qp_j^T \qB_{kj} \qp_j-(\qd_k^T\qp_k)^2+\sigma^2},
    \end{aligned}
\end{equation}
with
\vspace{-0.5em}
\begin{itemize}
 \item 
 $\qd_k = \big[d_{1k}, \cdots, d_{Mk}\big]^T \in \mathbb{R}^{M\times 1}$, 
with
$d_{mk}=\big|\mathbb{E}\{{\qg}_{m k}^H\qq_{m k}\}{\big|},  m \in \Mset \in \mathbb{R}^{M\times 1}$ . 
\item $\qp_k = [p_{1k}, \cdots, p_{Mk}]^T, $
with $p_{mk} =  a_{m k}\sqrt{\eta_{m k}},  m \in \Mset$.
\item $\qB_{kj} \in \mathbb{R}^{M \times M},$ with $[\qB_{kj}]_{l m} = \mathrm{Re}(\mathbb{E}\{\qg_{l k}^H \qq_{l j} \qq_{mj}^H\qg_{mk}\}),  l, m \in \Mset$.
\end{itemize}
In addition, a closed-form expression for the hardening-based pseudo-SE at user $k$ can be written as
\begin{align}\label{eq:SLINR_stat}
\tilde{\zeta}_{k}=\log _{2} (1+\widetilde{\SLINRk}),
\end{align}
where
\begin{equation}\label{eq_emil}
    \begin{aligned}
        \widetilde{\SLINRk} =\frac{(\tilde{\qd}_k^T\tilde{\qp}_k)^2}{\sum\limits_{j \in \Usm} \!\!\tilde{\qp}_j^T \tilde{\qB}_{kj} \tilde{\qp}_j\!+\!\!\sum\limits_{j \in \Unsm}\!\! \tilde{\qp}_k^T \tilde{\qF}_{kj} \tilde{\qp}_k-(\tilde{\qd}_k^T\tilde{\qp}_k)^2+\!\sigma^2},
    \end{aligned}
\end{equation}
\begin{itemize}
 \item 
 $\tilde{\qd}_k = \big[\tilde{d}_{1k}, \cdots, \tilde{d}_{|\Csm|k}\big]^T  \in \mathbb{R}^{|\Csm|\times 1}$,
with
$\tilde{d}_{ik}=\big|\mathbb{E}\{{\qg}_{\ell k}^H\qq_{\ell k}\}{\big|},  \ell=\Csm\{i\}$, $\forall i \in \{1, \cdots, |\Csm|\}$. 
\item $\tilde{\qp}_k = [\tilde{p}_{1k}, \cdots, \tilde{p}_{|\Csm|k}]^T \in \mathbb{R}^{|\Csm|\times 1}, $
with $\tilde{p}_{ik} =  a_{\ell k}\sqrt{\eta_{\ell k}},  \ell=\Csm\{i\}$, $\forall i \in \{1, \cdots, |\Csm|\}$. 
\item $\tilde{\qB}_{kj} \in \mathbb{R}^{|\Csm|\times |\Csm|},$ with $[\tilde{\qB}_{kj}]_{l m} = \mathrm{Re}(\mathbb{E}\{\qg_{\ell k}^H \qq_{\ell j} \qq_{oj}^H\qg_{ok}\}),  \ell=\Csm\{l\}, o=\Csm\{m\}$, $\forall l, m \in \{1, \cdots, |\Csm|\}$. 
\item $\tilde{\qF}_{kj} \in \mathbb{R}^{|\Csm|\times |\Csm|},$ with $[\tilde{\qF}_{kj}]_{l m} = \mathrm{Re}(\mathbb{E}\{\qg_{\ell j}^H \qq_{\ell k} \qq_{ok}^H\qg_{oj}\}),  \ell=\Csm\{l\}, o=\Csm\{m\}$, $\forall l, m \in \{1, \cdots, |\Csm|\}$.
\end{itemize}

\vspace{-0.5em}
\subsection{Problem Formulation and Modified WMMSE-Based Solution}
We now propose a statistical CSI-based resource allocation design for the given precoding. In this approach, resource allocation is updated according to the large-scale fading time scale (statistical channel properties).
For the optimization, any precoding design can be utilized, but in the simulation results we will consider cluster-wise MMSE precoding scheme.
Let $\widehat{\qG}_{s}$ be an $L|\Csm| \times |\Usm|$ collective channel estimation matrix for corresponding APs in set  $ \Csm$. More specifically, $\widehat{\qG}_{s}$ consists of  $|\Csm|\times|\Usm|$ vectors of dimension $L\times 1$, $\hat {\qg}_{ij}$, each corresponding to a particular AP $i$ in set  $ \Csm$ and user $j$ in set  $\Usm$ as
\begin{equation}\label{eq:Qij}
\widehat{\qG}_s=[\hat {\qg}_{ij}: i \in \Csm, j \in \Usm].
\end{equation}
By  using  MMSE precoding scheme the whole precoding vector constructed for the APs in $\Csm$ can be expressed as
\begin{equation}\label{eq:Q}
\widetilde{\qQ}_{s}=
\widehat{\qG}_{s} \left ({(\widetilde{\qP}_{s}\circ\widehat{\qG}_s)^{H}\widehat{\qG}_s }+\sigma^2 \qI_{|\Usm|}\right)^{-1},
\end{equation}
where, $\widetilde{\qP}_{s}=[\tilde {p}_{ij} \textbf{1}_{L\times 1}: i \in \Csm, j \in \Usm]$,  $\circ$ is the
Hadamard or entry-wise product.
Now, each BBH  construct its precoding vector for transmission to  user $k$ by choosing the corresponding  column vector.
In particular, let AP $m$ correspond to the  $i$-th element of set $\Csm$,  $i\in \{1,\cdots,|\Csm|\}$, and user $k$ correspond to the $j$-th element of set $\Usm$,  $j\in \{1,\cdots,|\Usm|\}$.  Accordingly, the downlink precoding vector constructed for the $m$-th AP to the $k$-th user can be calculated as the vector of $\tilde{\qQ}_{s}$ obtained by selecting the $L$ rows $a_{mk}$ to $\check{a}_{mk}$ from $j$-th column as $\qq_{mk}=\frac{\tilde{\qq}_{mk}}{\|\tilde{\qq}_{mk}\|_2}$ with
\vspace{-0.5em}
\begin{align}\label{eq:Qmk}
\tilde{\qq}_{mk}&= \big[{\widetilde{\qQ}_{s}}\big]_{(a_{mk}: \check{a}_{mk}, j)},
\end{align}
where   $a_{mk}=(i-1)\times L+1$ and $\check{a}_{mk}=a_{mk}+ L-1$.

Now, for  PC $\Csm$ we formulate the following hardening-based sum pseudo-SE maximization  problem
\begin{subequations}\label{eq:problemstat1}
\begin{align}
\max _{\boldsymbol {\eta}_s, \aset} &\sum\nolimits_{k \in \Usm} \tilde{w}_k  \tilde{\zeta}_{k}(\boldsymbol {\eta }_s,\aset) \label{eq:problemstat1:O1}
\\
&\hspace{-1em}\text {st.}\hspace{1em}\eta _{mk}\geq 0,\hspace{2em} \forall m\in \mathcal{C}_s,~k\in\mathcal{U}_s, \label{eq:problemstat1:C1} 
\\
&\hphantom {\text {st.}} \sum\nolimits_{k \in \Usm} a_{mk}\leq \Kmax, \quad \forall m\in \mathcal{C}_s, \label{eq:problemstat1:C3} 
\\
&\hphantom {\text {st.}}  \sum\nolimits _{k\in \Usm} \tilde{p}_{mk}^{2} \leq \Pb, ~\quad  \forall m\in \mathcal{C}_s. \label{eq:problemstat1:C4}
\end{align}
\end{subequations}
where $\tilde{w}_k$ presents the priority weight associated with user $k$.
Following the same approach as in the previous section, we can formulate the equivalent WMMSE problem to optimize power control and user association as
\begin{subequations}
   \begin{align} \label{eq:WMMSEstat2}
\mathop {\mathrm {min\kern 0pt}} _{\left \{\tilde{\rho} _{k}, u_{k}, \tilde{\qp}_{k}|  k\in \Usm\right \}} &\sum\nolimits _{k\in \Usm} \tilde{w}_{k} \left (\tilde{\rho} _{k} \tilde{e} _{k} - \log \tilde{\rho} _{k} \right ) 
\\ 
&\hspace{-2em}~ \mathop {\mathrm {st.\kern 0pt}}\nolimits \quad\!\! \sum\nolimits _{k\in\Usm}  \tilde{p}_{mk}^{2} \leq \Pb, ~\quad \forall m \in \Csm,
\\ 
& \sum\nolimits _{k\in \Usm} \tilde{\vartheta} _{mk}  \tilde{p}_{mk}^{2} \leq \Kmax, ~\forall m \in \Csm,
\end{align}  
\end{subequations}
where 
\begin{align}\label{eq:varho_stat}
        \tilde{\vartheta} _{mk} = \big({\tilde{p}_{mk}^2+ \epsilon}\big)^{-1},
\end{align}
and
\begin{align}\label{eq:e_stat}
    \tilde{e}_k = &  \tilde{u}_k ^2\bigg(\sum\nolimits_{j \in \Usm} \!\!\tilde{\qp}_j^T \tilde{\qB}_{kj} \tilde{\qp}_j\!+\!\sum\nolimits_{j \in \Unsm}\!\! \tilde{\qp}_k^T \tilde{\qF}_{kj} \tilde{\qp}_k +
    \sigma^2\bigg) \nonumber\\
    &
    - 2\tilde{u}_{k} \tilde{\qd}_k^T\tilde{\qp}_{k} + 1. 
    \end{align}

For fixed $\tilde{\qp}_{k}$ and $\tilde{\rho}_{k}$, the objective function in~\eqref{eq:WMMSEstat2} can be minimized with respect to $\tilde{u}_k$ by setting its first-order derivative to zero, resulting in
\begin{equation}
\label{u_stat}
    \tilde{u}_k^\opt = \frac{\tilde{\qd}_k^T\tilde{\qp}_{k}}{\sum\nolimits_{j \in \Usm} \!\!\tilde{\qp}_j^T \tilde{\qB}_{kj} \tilde{\qp}_j\!+\!\sum\nolimits_{j \in \Unsm}\!\! \tilde{\qp}_k^T \tilde{\qF}_{kj} \tilde{\qp}_k+\sigma^2}.
\end{equation}
Accordingly, for fixed $\tilde{u}_{k},\tilde{\qp}_{k}$, the optimal MSE weight is
\begin{align}
\label{eq:rho_stat}
    \tilde{\rho}^\opt_{k} = (\tilde{e}_k)^{-1}.
\end{align}
 Finally, with fixed $\tilde{u}_k$, $\tilde{\rho}_k$, we obtain a QCQP problem with respect to $\tilde{\qp}_{k}$ as
 \begin{subequations}
      \begin{align}\label{eq:WMMSEstat3}
 \mathop {\mathrm {min\kern 0pt}} ~  &\tilde{f} (\tilde{\qp})
 \\
 &\hspace{-1.5em} {\mathrm {st.\kern 0pt}}\hspace{0.5em} \sum\nolimits _{k\in\Usm}  \tilde{p}_{mk}^{2} \leq \Pb, ~\quad \forall m \in \Csm,
 \\
& \sum\nolimits _{k\in \Usm} \tilde{\vartheta} _{mk}  \tilde{p}_{mk}^{2} \leq \Kmax, ~\quad \forall m \in \Csm,
\end{align} 
 \end{subequations}
where 
\begin{align}
    \tilde{f}(\tilde{\qp}) =& \sum_{k \in \Usm}\tilde{w}_k\tilde{\rho}_k\tilde{u}_k^2\bigg(\sum_{j \in \Usm} \!\!\tilde{\qp}_j^T \tilde{\qB}_{kj} \tilde{\qp}_j\!+\!\sum_{j \in \Unsm}\!\! \tilde{\qp}_k^T \tilde{\qF}_{kj} \tilde{\qp}_k +
    \sigma^2\bigg)\nonumber\\
    &- 2\tilde{w}_k\tilde{\rho}_k\tilde{u}_{k} \tilde{\qd}_k^T\tilde{\qp}_{k}
\end{align}
and $\tilde{\qp}\triangleq \{\tilde{\qp}_{k}|  k\in \Usm\}$. 

Both the objective function and the constraints in \eqref{eq:WMMSEstat3} are convex and standard convex optimization tools can be used to solve it. The solution of~\eqref{eq:WMMSEstat3}  is
shown in  \textbf{Algorithm \ref{Alg:WMMSE2}}.

\begin{algorithm}[t]\label{Alg:WMMSE2}
\caption{Statistical CSI-Based  Sum-Pseudo-SE Maximization With Modified WMMSE Approach at PC $\Csm$}\label{Alg:WMMSE2}
\begin{algorithmic}[1]
\State \textbf{Initialize}: $\tilde{\vartheta}_{mk}^{(0)}$, $\tilde{\qp}_{k}^{(0)}, \forall k \in \mathcal{U}_s, \forall m \in \mathcal{C}_s$, iteration index $i = 0$, convergency
accuracy $\xi$.
\While  {$\frac{\big|\sum_{k \in \mathcal{U}_s} \tilde{w}_k  \tilde{\zeta}_{k}^{(i)}-\sum_{k \in \mathcal{U}_s}  \tilde{w}_k  \tilde{\zeta}_{k}^{(i-1)}\big|} {\sum_{k \in \mathcal{U}_s}  \tilde{w}_k  \tilde{\zeta}_{k}^{(i-1)}}<\xi$}
\State $i = i + 1$;
\State Calculate $\tilde{u}_k^{(i)}$  according to \eqref{u_stat} with $\tilde{\qp}_{k}$ fixed, $k\in \Usm$
\State Calculate $\tilde{e}_k^{(i)}$ according to \eqref{eq:e_stat} with $\tilde{\qp}_{k}$ and $\tilde{u}_k$ fixed, $k\in \Usm$.  
\State   Update $\tilde{\rho}_k$  according to \eqref{eq:rho_stat}.
\State Calculate the optimal power control $\tilde{\qp}_{k}$ with  $\tilde{u}_k$ and $\tilde{\rho}_k$ fixed $\forall k\in \Usm, m \in \Csm$, by solving the problem~\eqref{eq:WMMSEstat3}
\State Calculate $\tilde{\vartheta}_{mk}^{(i)}$ according to \eqref{eq:varho_stat}
 \EndWhile
\end{algorithmic}
\end{algorithm}



\vspace{-0.5em}
\section{Numerical Results} \label{Sec:numerical}
We consider a CF-mMIMO system where the APs and users are randomly distributed within a $1 \times 1$ km${}^2$, with wrapped-around edges to eliminate boundary effects. In cluster-wise processing, the available APs are divided into $S$ disjoint PCs, which can be based on factors such as AP location and interference relationships. 
In our simulations, we group APs into PCs based on their geographical locations. Moreover, for each user $k$, we choose a  cluster $\Csm$ comprising   the APs providing the highest sum average received power at user $k$. 
In particular, we select cluster $\Csm$ for user $k$ with 
\begin{align}
s =   \underset{s \in \{1,\cdots, S\}}\arg \max \sum\nolimits_{m \in \Csm} \beta_{mk}.
\end{align}
Each AP can serve up to $K_{\mathrm{max}}$ users out of a set of $K$ users. The value of parameter $K_{\mathrm{max}}$  depends on the system and fronthaul parameters and is determined based on~\eqref{eq:Kmax}.  
We evaluate the performance of  CF-mMIMO system with cluster-wise processing   under limited fronthaul capacity   and   transmit power constraint  relying on our proposed   Algorithm 1 and Algorithm 2. We compare the performance of the following cases:
\begin{itemize}
\item Network-wide (centralized) processing with $S=1$:   In this case, there is one PC consisting of all the APs in the network. The performance of this case can be considered the fronthaul-limited upper bound.

\item Cluster-wise (decentralized) processing with  $S>1$ PCs: In this case, APs are divided into $S \in \{2, 4\}$ PCs. In each PC, we resort to the proposed instantaneous CSI-based Algorithm 1 (statistical CSI-based Algorithm 2) to design precoding, power allocation, and user association (power allocation and user association).
\end{itemize}

\vspace{-0.5em}
\subsection{ Parameters and Setup}
The maximum transmit power for training  pilot sequences and for the transmit power at each AP is set to $100$ mW. The noise power is $\sigma^2_n=-92$ dBm, while the  fronthaul parameters are chosen based on  Table~\ref{tab:FronthaulParameter}.  In addition,  we consider $100$ MHz bandwidth with $30$ kHz subcarrier spacing which corresponds to $N_{\mathrm{sub}}=3264$  as in Table~\ref{tab:FronthaulParameter}. In addition, we set $\tauu = 2000$ samples, which corresponds to a coherence bandwidth of $200$ KHz and a coherence time of $10$ ms.
The large-scale fading  and the path-loss between AP $m$ and user $k$
is modeled as
\begin{equation} \label{eq:Beta}
\beta_{mk} =\text{PL}_{mk}  10^{\frac {\sigma _{sh} ~y_{mk}}{10}}, 
\end{equation}
where $\text{PL}_{mk}$ denotes the path loss and 
$10^{\frac {\sigma _{sh} ~y_{mk}}{10}}$ represents the shadow fading  with standard deviation $\sigma_{sh} = 4$ dB and $y_{mk} \sim \mathcal{CN}(0, 1)$. To model $\text{PL}_{mk}$, we consider the popular three-slope model
as described in\cite{hien:2017:wcom}.   
\begin{figure*}\label{fig1}
\centering 
\begin{minipage}{.48\textwidth} 
\centering 
\includegraphics[width=0.9\linewidth]{  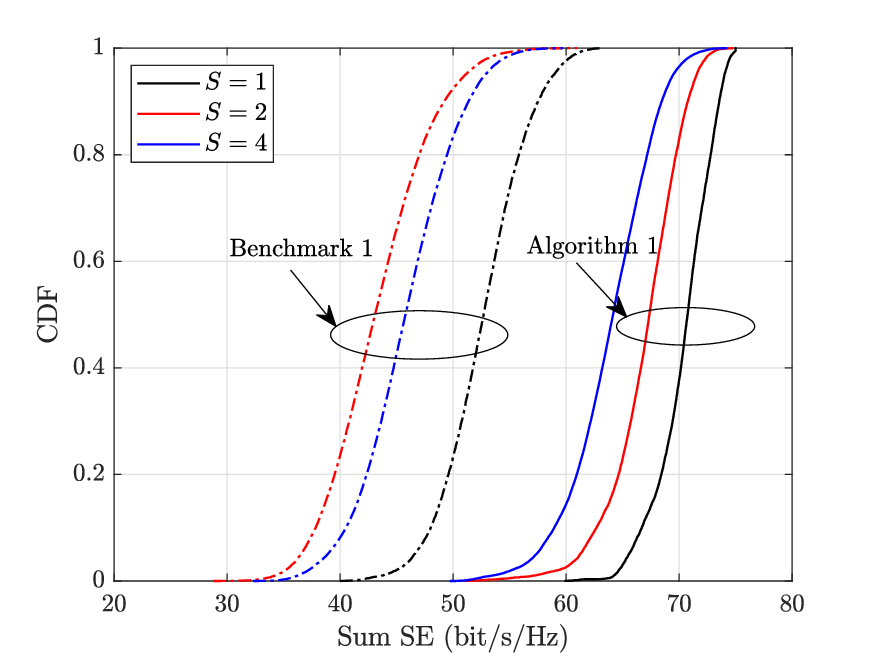}
\vspace{-0.5em}
\subcaption[a]{Instantaneous CSI-based design Algorithm 1}
\label{fig1_a} 
\end{minipage}%
\begin{minipage}{.48\textwidth} 
\centering \includegraphics[width=0.9\linewidth]{  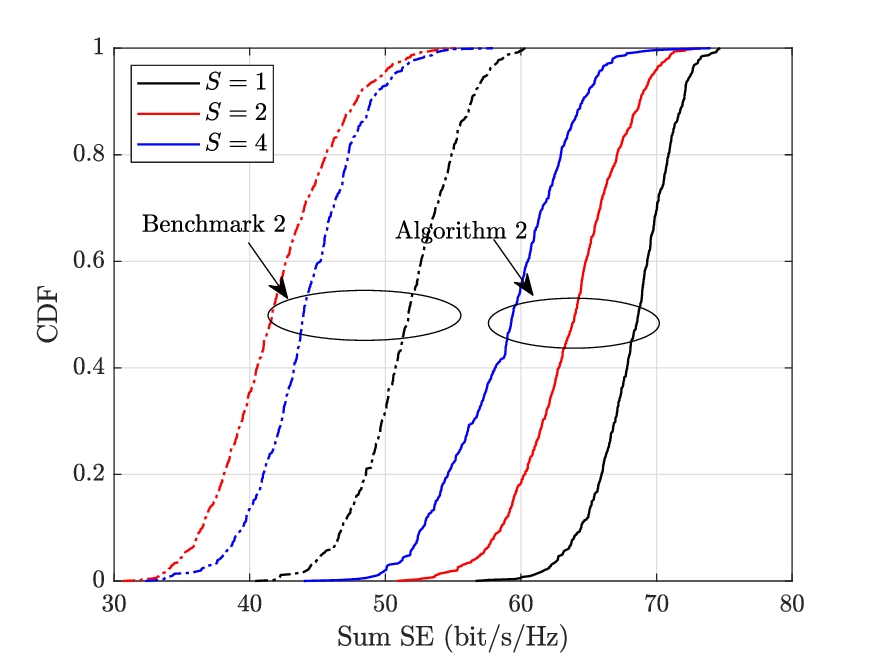}
\vspace{-0.5em}
\subcaption[b]{Statistical CSI-based design Algorithm 2}
\label{fig1_b}
\end{minipage} 
\caption{  Comparison among the sum-SE achieved by the proposed
Algorithms and benchmark schemes, where $L=24$, $K=15$,  $M=10$, $\mathrm{FH_{max}}=10$ Gbps, and $M_{\mathrm{mo}}=32$.
}\label{fig1}
\vspace{-1.2em}
\end{figure*}

\begin{table}[t]
	\centering
\caption{  Fronthaul Parameters} 
 \begin{tabular}{||c c c c ||}
 \hline
 Parameter & Value & Parameter & Value \\ [0.5ex]
 \hline\hline
 $N_{\mathrm{sub}}$ & 3264  &    {$N_{\mathrm{Gran}}$}       &136  \\
 \hline
$\varepsilon_\mathrm{{cp}}$ & 0.85      &$\delta_{\mathrm{pr}} (\delta _\mathrm{da})$    & $2\times10^{-4}$($5\times10^{-4}$) s  \\
 \hline
$N_{\mathrm{o}}$ & 14 sym &$N_{\mathrm{bits}}$ & 16 \\
 \hline
\end{tabular}
\label{tab:FronthaulParameter}
\end{table}

\subsection{Results and Discussions}
 \subsubsection{ Performance of the Proposed   Cluster-wise  Weighted Pseudo-SE-Maximization Approaches}

In Figs.~\ref{fig1_a} and~\ref{fig1_b}   we evaluate the performance of
the proposed cluster-wise  pseudo-SE-maximization approaches in Algorithm 1 and Algorithm 2 in the fronthaul-aware   CF-mMIMO system, respectively. We consider the following benchmark schemes for comparisons:
\begin{itemize}
\item Benchmark 1: In this scheme, we consider successive instantaneous CSI-based cluster-wise processing design, where user association, precoding, and  power allocation  are successively implemented over the
small-scale fading time scale. In particular, for user association,  each AP sorts the instantaneous channel gains in descending order and independently selects $\Kmax$ users with the strongest channel gains.For the given user association,  the power control coefficients are determined based on the low-complexity heuristic scheme proposed in~\cite{nayebi:2017:wcom}, which has shown excellent performance and serves as a reliable baseline. Using this scheme, the power coefficient used by AP $m$ for transmission to user $k$ is calculated as
 \begin{equation}~\label{EPA:Benchmark}
\eta_{mk}=\frac{1}{\max_m\left(\sum_{k \in \Kset}a_{mk}\right)}, \quad m \in \Csm,~k\in \Usm. 
\end{equation}

In addition, precoding vectors are designed based on   cluster-wise MMSE scheme.   
\item Benchmark 2: In this scheme, we consider successive statistical CSI-based cluster-wise processing design, where user association  and  power allocation are successively implemented over the
large-scale fading time scale. In particular, for user association,  each AP $m$ sorts the large-scale fading coefficients, $\beta_{mk}$, in descending order and independently selects $\Kmax$ users with the strongest channel gains.  For power control, given the user association, we use heuristic scheme~\cite{nayebi:2017:wcom}, while precoding vectors are designed based on~\eqref{eq:Qmk} for the given user association and power control.
\end{itemize}

The main observations that follow from these simulations are
as follows:
\begin{itemize}
\item The  proposed cluster-wise processing solutions enhance the system performance significantly for both the statistical and instantaneous CSI-based designs. In particular, when $S=2$, the joint precoding, power allocation, and user association Algorithm 1 provides a performance gain of up to $56\%$ compared to \emph{Benchmark 1}. Meanwhile, the statistical CSI-based design  in Algorithm 2  yields a performance gain of up to $52\%$ compared to \emph{Benchmark 2}. This highlights the advantage of our proposed solutions over the heuristic benchmarks.
 
\item The performance gap between network-wide CF-mMIMO with $S=1$ and CF-mMIMO with cluster-wise processing,  decreases with our proposed solutions in Algorithm 1 and Algorithm 2. More precisely,  the performance loss of CF-mMIMO with statistical CSI-based cluster-wise  processing compared  to the  centralized  case is around $7\%$ and $16\%$ when $S=2$ and $S=4$, respectively. These losses reduce to  $4\%$ and $9\%$, for instantaneous CSI-based design, respectively. 
This is an interesting result because it shows the
importance of deploying multiple PCs in the CF-mMIMO system.
\end{itemize}

Figure~\ref{fig:Fig1c} compares the performance of the joint optimization approach in Algorithm 1 for the CF-mMIMO system with $S=2$ PCs against cases where only the power allocation coefficients or both power allocation and user association variables are optimized, denoted by OPA and OPA-OUA, respectively. It is observed that OPA yields a $28\%$ performance gain over Benchmark 1, and OPA-OUA provides a $15\%$ improvement over OPA by optimizing both power allocation and user association. Also, Algorithm 1 achieves an additional $8\%$ gain over OPA-OUA by jointly optimizing user association, power allocation, and precoding. This demonstrates that the integrated optimization approach in Algorithm 1 significantly outperforms the individual optimizations.


\begin{figure}[t]
	\centering
	\includegraphics[width=0.43\textwidth]{  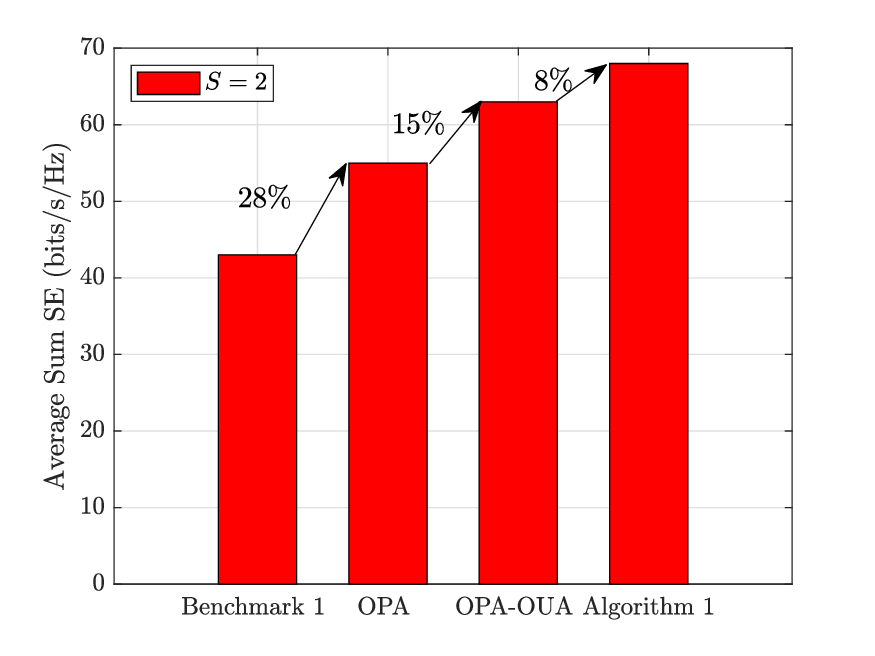}
	\vspace{-1em}
	\caption{Comparison among the sum-SE achieved by different optimization approaches, where $L=24$, $K=15$,  $M=10$, $\mathrm{FH_{max}}=10$ Gbps, and $M_{\mathrm{mo}}=32$.  }
	\vspace{-0.2em}
	\label{fig:Fig1c}
\end{figure}

\subsubsection{ Impact of the Number of Antennas per AP} Figure~\ref{fig:Fig2} presents the average sum-SE performance of the fronthaul-limited CF-mMIMO system with cluster-wise processing as a function of the number of transmit antennas at the AP. The main insights drawn from this figure are as follows.
\begin{itemize}
\item Increasing  the number of transmit antennas at each AP impacts the sum-SE performance in two ways: (i) it boosts diversity and array gain, and (ii) reduces $K_{\mathrm{max}}$ due to fronthaul limitations~\eqref{eq:Kmax}. However, the first effect dominates and results in a notable enhancement in sum-SE performance, especially for the statistical CSI-based design with higher number of PCs. 

\item It is evident that for a large number of transmit antennas, CF-mMIMO employing 
 scalable cluster-wise processing along with the proposed WMMSE-based solutions incur only a minor  performance loss for decentralization, e.g., around $10\%$ and $12\%$ for instantaneous and statistical CSI-based designs with $S=4$ PCs, respectively. This is achieved while significantly reducing the computational complexity of cluster-wise processing and the overhead required for CSI acquisition across large processing sets. These results emphasize the benefit of our proposed scalable cluster-wise processing, namely minimizing the need for extensive network-wide processing and coordination among numerous APs, while still providing competitive performance.
 \item Herein, we also present the results for
CF-mMIMO having infinite fronthaul capacity (represented as infinite FH) relying on the proposed WMMSE-based algorithms when $S=1$, $K_{\mathrm{max}}=K$,  and there is no post-imposing fronthaul constraint~\eqref{eq:FHMorder}, i.e.,  $R_{k, \mathrm{post}} = R_{k}$. It is observed that the fronthaul constraints lead to the performance loss, which is more pronounced for the CF-mMIMO systems in the regime of large values of $L$. Nevertheless,  our proposed  optimization solutions in Algorithm 1 and 2  could potentially make  CF-mMIMO with cluster-wise processing  competitive
compared to network-wide CF-mMIMO with infinite fronthaul capacity.
\end{itemize}

\begin{figure}[t]
	\centering
	\includegraphics[width=0.43\textwidth]{  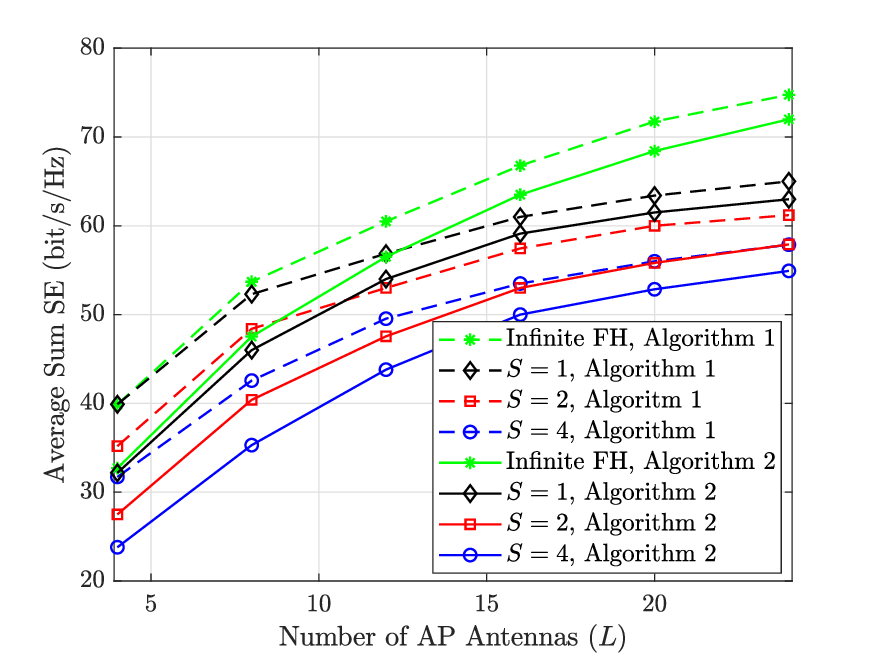}
	\vspace{-0.5em}
	\caption{Average sum SE  versus number of AP antennas, $L$, where   $K=15$,  $M=8$, $\mathrm{FH_{max}}=10$ Gbps, and $M_{\mathrm{mo}}=32$.  }
	\vspace{-0.5em}
	\label{fig:Fig2}
\end{figure}

\begin{figure}[t]
	\centering
	\includegraphics[width=0.43\textwidth]{  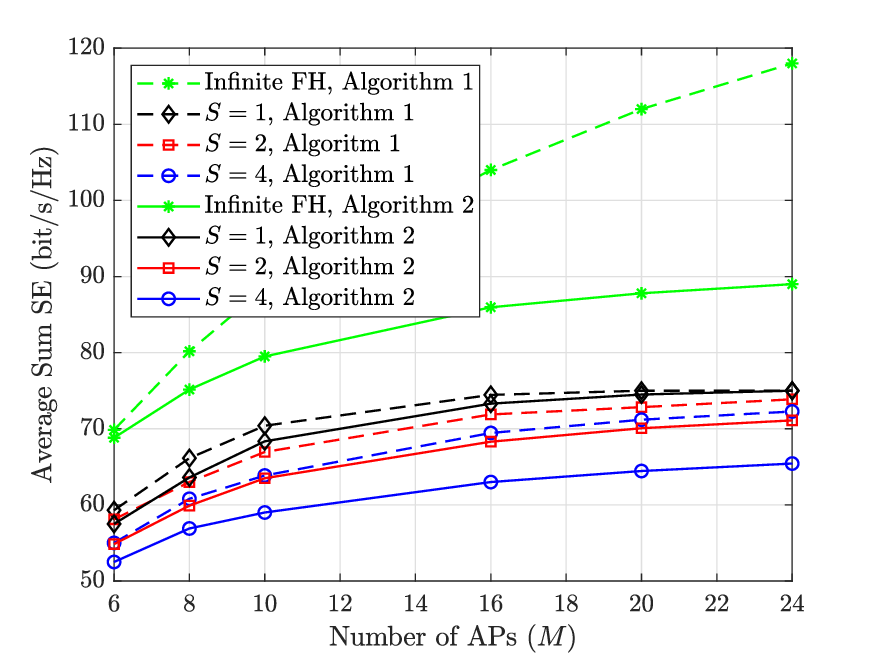}
	\vspace{-0.4em}
	\caption{Average sum SE  versus number of APs, $M$, where $ML=240$,   $K=15$,    $\mathrm{FH_{max}}=10$ Gbps, and $M_{\mathrm{mo}}=32$.}
    \vspace{-0.2em}
	\label{fig:Fig3}
\end{figure}
\subsubsection{ Impact of the Number of APs}
Figure~\ref{fig:Fig3} presents the  average sum SE achieved by CF-mMIMO system  for  different numbers of   APs for systems having the same total numbers of service antennas, i.e., $L M = 240$, but different
number of APs. The main observations that follow from these simulations are as follows.
\begin{itemize}
    \item For all the schemes, distributing antennas  results in better sum-SE performance due to the additional macro diversity gain.

    \item   It can be observed that, as the number of APs increases, the sum-SE performance of the fronthaul-limited CF-mMIMO system remains relatively unchanged when $M \geq 20$,    particularly for the $S=1$ and $S=2$ cases. This behaviour contrasts with the infinite fronthaul scenario, where the SE increases notably with $M$. The limited improvement in the fronthaul-limited case is due to the fronthaul constraint in~\eqref{eq:FHMorder}, which doesn't allow the SE to increase more than $\log_2(M_{\mathrm{mo}})$. As a result, the fronthaul bottleneck prevents further SE gains as $M$ increases. 
    For this regime the performance gap between network-wise processing and  cluster-wise processing under fronthaul constraint significantly reduces and cluster-wise processing is undoubtedly a better choice.
 
    \item  The performance of CF-mMIMO system with cluster-wise processing relying on statistical CSI-based design Algorithm 2 is fairly close to CF-mMIMO system relying on instantaneous
CSI-based design   in Algorithm 1 for the large to medium range of number of antennas for different number of PCs. For example, when  $L = 40$  or equivalently $M=6$, the performance gap between  statistical CSI-based and instantaneous CSI-based design is less than  $5\%$ and $7\%$ when there is $S=2$ PCs, and $S=4$ PCs, respectively. This behaviour follows  from the fact that the level of channel hardening remarkably increases for higher values of $L$ and hence,  using the mean of the effective gain instead of the true channel gains for cluster-wise processing  in CF-mMIMO systems    works very  well.  Even when it comes to the network-wide processing, the statistical CSI-based design in the fronthaul-limited  CF-mMIMO system is capable to achieve  $97\%$ of the average sum-SE with instantaneous CSI-based design.  
It is an interesting observation, since statistical CSI-based designs provide better trade-offs between performance, complexity, and signaling overhead.
 \end{itemize}

\begin{figure}[t]
	\centering
	\includegraphics[width=0.43\textwidth]{  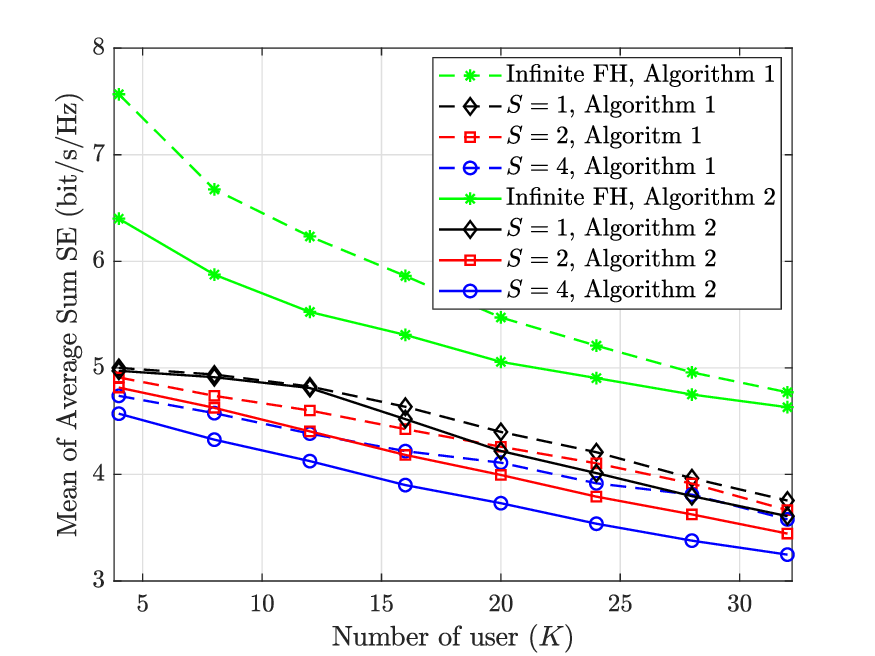}
	\vspace{-0.3em}
	\caption{{Mean of average sum SE  versus number of users, $K$, where $L=24$,    $M=10$, $\mathrm{FH_{max}}=10$ Gbps, and $M_{\mathrm{mo}}=32$.  }}
	\label{fig:Fig4}
\end{figure}
\subsubsection{ Impact of the Number of Users }
Figure~\ref{fig:Fig4} shows the mean of average sum-SE, i.e.,  (average sum-SE)$/K$, of a fronthaul-limited CF-mMIMO system, as a function of the number of users for different number of PCs. We observe that in the regime of small values of  $K$, the performance gap between all the schemes are very  small. On the other hand, by increasing $K$, the sum-SE performance of all cases deteriorates. Nevertheless, the CF-mMIMO system using the cluster-wise processing   still yields excellent SE performance compared to network-wide processing. More specifically, when $K = 28$, the performance loss of cluster-wise processing with $S=2$ ($S=4$)  PCs compared to the network-wide processing is less than $1.2\%$ ($3\%$) for instantaneous CSI-based design. Importantly, this is the case with statistical CSI-based design; the performance loss of cluster-wise processing with $S=2$ ($S=4$)  PCs  is less than $3\%$ ($5\%$). These negligible performance losses verify the importance of an adequate cluster-wise  processing along with statistical CSI-based design to provide a better
performance/implementation complexity trade-off compared to its instantaneous CSI-based network-wide processing counterpart.

\begin{figure}[t]
	\centering
	\includegraphics[width=0.43\textwidth]{  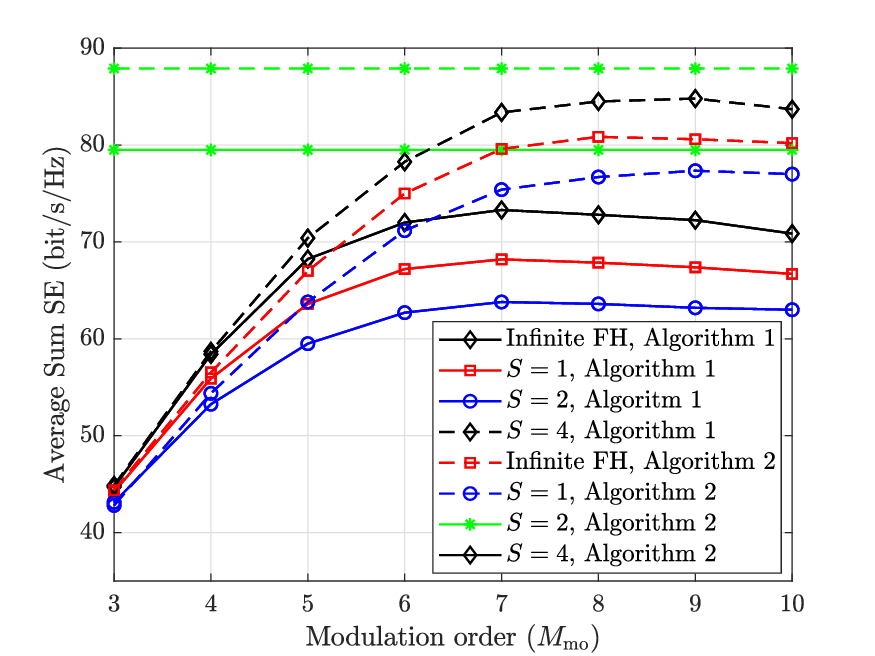}
	\vspace{-0.5em}
	\caption{Average sum SE  versus modulation order, $M_{\mathrm{mo}}$, where $L=24$, $K=15$,  $M=10$,  $\mathrm{FH_{max}}=10$ Gbps.}
	\vspace{-0.2em}
	\label{fig:Fig5}
\end{figure}
\subsubsection{ Effect of the Modulation Order}
In Fig.~\ref{fig:Fig5}    we investigate the  average sum SE performance of the CF-mMIMO system with   cluster-wise processing architecture for different number of PCs as a function of $M_{\mathrm{mo}}$.  It is observed that there exists an optimal value of $M_{\mathrm{mo}}$ for each scheme that maximizes the average sum SE performance. This is reasonable because, on one hand, the fronthaul consumption for transmitting information symbols increases with $M_{\mathrm{mo}}$, which reduces $\Kmax$. On the other hand, a higher $M_{\mathrm{mo}}$ enables higher SE under the fronthaul constraint~\eqref{eq:FHMorder}. Therefore, there is a trade-off between the SE and $M_{\mathrm{mo}}$. In addition, the sum-SE performance gap for the systems relying on instantaneous CSI-based design and statistical CSI-based design increases with higher $M_{\mathrm{mo}}$, while the performance gap between cluster-wise processing schemes and network-wide processing is  relatively small when $M_{\mathrm{mo}} \leq 4$. Therefore, for the application scenarios with lower $M_{\mathrm{mo}}$,  cluster-wise processing relying on statistical CSI-based design is undoubtedly a better choice.

\begin{figure}[t]
	\centering
	\includegraphics[width=0.43\textwidth]{  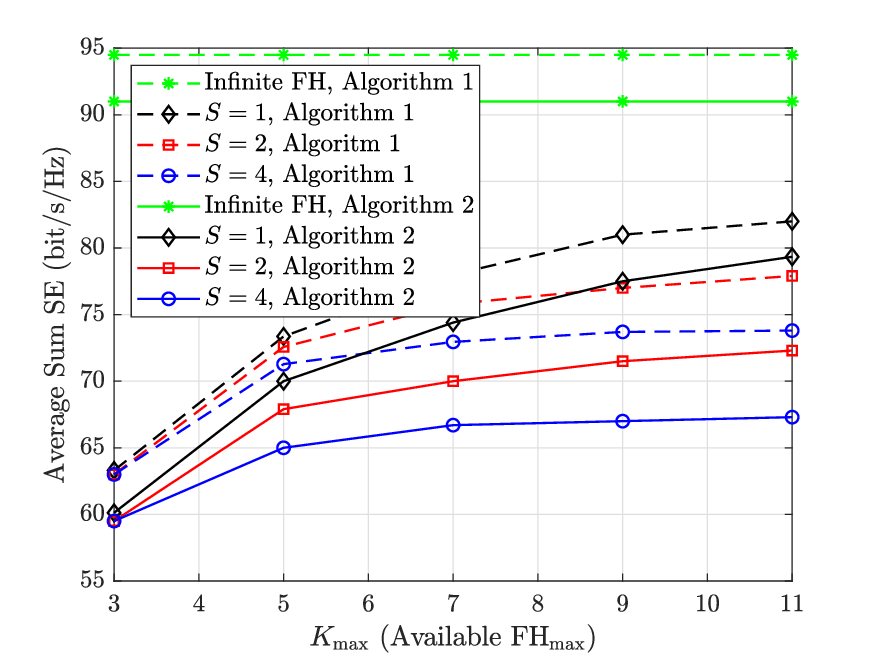}
	\vspace{-0.5em}
	\caption{Average sum SE  versus $\Kmax$ ($\mathrm{FH}_{\mathrm{max}}$), where $L=24$, $K=20$,  $M=8$,    $M_{\mathrm{mo}}=32$.}
	\vspace{-0.2em}
	\label{fig:Fig6}
\end{figure}

\subsubsection{ Impact of the Available Fronthaul Capacity}  Figure~\ref{fig:Fig6} illustrates   the effect of the maximum available fronthaul capacity, $\mathrm{FH}_{\mathrm{max}}$,  on the sum-SE performance of the  CF-mMIMO system with the  proposed   cluster-wise processing.  Different values of $\mathrm{FH}_{\mathrm{max}}$, correspond to different values of $\Kmax$, which are calculated based on~\eqref{eq:Kmax}.
It is observed that the limited fronthaul capacity reduces the system  performance.  For example, under the instantaneous CSI-based design, when the fronthaul capacity $\mathrm{FH}_{\mathrm{max}}=6$ Gbps (or equivalently $\Kmax=5$), there is $22\%$ performance loss due to fronthaul limitation for the centralized scheme with $S=1$ compared to the case of infinite fronthaul links. This loss  slightly increase to $23\%$ when $S=2$.  On the other hand, we observe that upon increasing  $\Kmax$ the sum-SE performance of all cases increases, specially for centralized schemes.
 Simulation results also confirm that multiple  clusters  is better suited for  CF-mMIMO network architectures with low-capacity  fronthaul links.

{Finally, we would like to emphasize that fairness among users is inherently promoted through     our optimization framework. Specifically, we adopt a proportional fairness strategy by weighting each user's  pseudo-SE  with a priority coefficient \( w_k \), which is chosen as the inverse of the average achievable SE experienced by that user. For example, in the Monte Carlo evaluation of the instantaneous CSI-based \textbf{Algorithm 1}, \( w_k \) at small-scale fading realization \( n \) is computed as the inverse of the user’s average   achievable SE up to realization \( n - 1 \).
Figure~\ref{fig:Fig9} presents the  CDF  of the per-user SE achieved by \textbf{Algorithm 1}. The 5th percentile SE (i.e., the SE value that 95\% of users exceed) is reported as $1.7, 1.4,$ and $1$ bit/s/Hz for  $S = 1, 2, 4$, respectively. These results indicate that even the least-served users maintain a non-negligible SE, thereby confirming that our approach achieves a desirable balance between efficiency and fairness.
}

\begin{figure}[t]
	\centering
	\includegraphics[width=0.43\textwidth]{  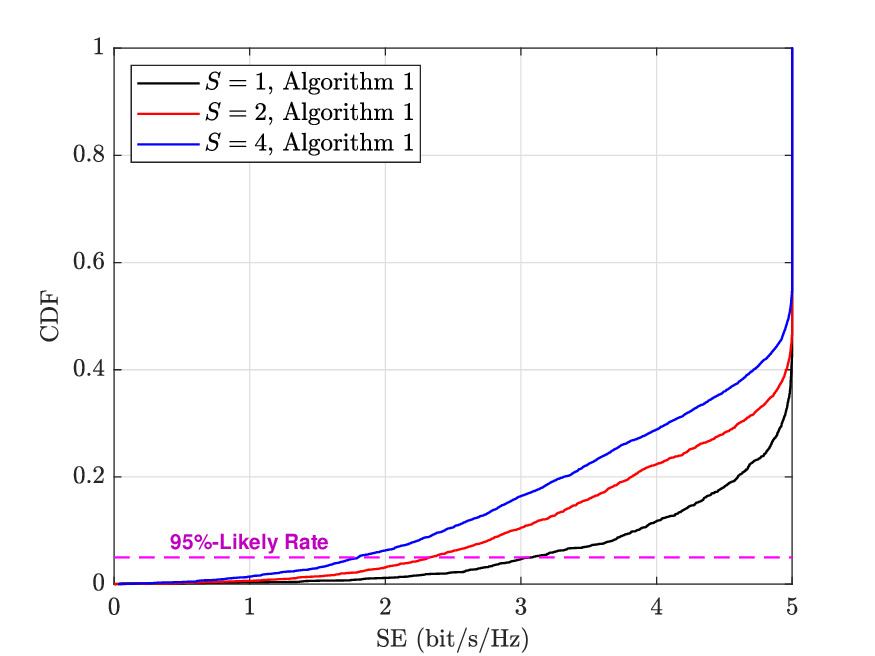}
	\vspace{-0.5em}
	\caption{ {CDF of the per-user SE achieved by Algorithm 1,  where $L=24$, $K=15$,  $M=10$, $\mathrm{FH_{max}}=10$ Gbps, and $M_{\mathrm{mo}}=32$. }}
	\vspace{-0.5em}
	\label{fig:Fig9}
\end{figure}
 {\subsection{ Computational Complexity }
Here, we discuss the complexity of the  proposed instantaneous CSI-based cluster-wise design  in \textbf{Algorithm 1} and statistical CSI-based cluster-wise design in \textbf{Algorithm 2}. We note that in the instantaneous CSI-based   design, both precoding and resource allocation for each PC are re-calculated based on the small-scale fading time scale (instantaneous channel conditions). It is important to note that small-scale fading coefficients fluctuate rapidly across time and frequency. Therefore, the required channel acquisition and computational complexity of \textbf{Algorithm~$1$} become prohibitive as the network size increases. However, in statistical CSI-based \textbf{Algorithm~$2$}, resource allocation is updated according to the large-scale fading time scale (statistical channel properties). Large-scale fading coefficients remain constant across frequencies and vary much more slowly over time compared to small-scale fading. Consequently,  with the statistical CSI-based designs, we need to perform \textbf{Algorithm~$2$} only $1$ time, and use the results for all subcarriers and several frame duration~\cite{Zahra:2025:LectureNote}.  
To quantify these differences,   the considered CF-mMIMO system with $100$ MHz bandwidth and $136$ kHz subcarrier spacing yields approximately $735$ subcarriers. Assuming $10$ transmit time intervals (TTIs) per 10 ms frame, \textbf{Algorithm~\ref{Alg:WMMSE1}} must be executed $735 \times 10 = 7350~\text{times per frame}$.
In contrast, \textbf{Algorithm~\ref{Alg:WMMSE2}} is executed \textbf{once}, dramatically reducing overhead.}

{Moreover, both algorithms are implemented in a  parallel cluster-wise manner, with each PC computing independently. Table~\ref{tab:runtime} summarizes the observed  average number of iterations  and  runtime per PC  for Algorithm~\ref{Alg:WMMSE1}. The observed reduction in runtime as $S$ increases confirms the  scalability  of the cluster-wise architecture.
Finally, in what follows, we calculate the computational complexity per PC $\Csm$ and per  iteration $i$.
 The computational complexity of Step 4 and Step 5 in \textbf{Algorithm 1} is $\mathcal{O}(|\Usm|^2|\Csm| L)$, while the  computational complexity for Step 6 is   $\mathcal{O}(|\Usm|)$. Step 7 of \textbf{Algorithm 1} involves solving a QCQP problem, which can be equivalently reformulated as a second-order cone programming (SOCP) problem. 
As discussed in~\cite{Zhou:JSP:2020}, the complexity of solving an SOCP problem   is $O(N_{\mathrm{so}} M_{\mathrm{so}}^{3.5}+ N_{\mathrm{so}}^3 M_{\mathrm{so}}^{2.5})$, where $M_{\mathrm{so}}$  is the number of second order cone constraints and $N_{\mathrm{so}}$ is the dimension of each.
Problem~\eqref{eq:WMMSE2} contains $|\Csm|$ transmit power constraints   and $|\Csm|$ fronthaul constraints with dimension $L|\Usm|$. Therefore, the complexity of solving Problem~\eqref{eq:WMMSE2} 
is $O(L |\Usm||\Csm|^{3.5}+ L^3|\Usm|^3 |\Csm|^{2.5})$.  Accordingly, the total computational complexity of \textbf{Algorithm \ref{Alg:WMMSE1}} per  iteration is $O(L |\Usm||\Csm|^{3.5}+ L^3|\Usm|^3 |\Csm|^{2.5})$.
The computational complexity of \textbf{Algorithm \ref{Alg:WMMSE2}} \emph{per iteration} is the same as that of \textbf{Algorithm \ref{Alg:WMMSE1}}.}
\begin{table}[t]
\centering
\caption{{Average number of iterations and runtime per PC versus number of PCs, $S$.}}
{\begin{tabular}{|c|c|c|}
\hline
\textbf{$S$} & \textbf{Avg. Number of Iterations} & \textbf{Avg. Runtime (sec)} \\
 \hline\hline
1 & 6.0 & 141 \\
 \hline 
2 & 4.5 & 45 \\
 \hline
4 & 3.5 & 9.5 \\
\hline
\end{tabular}}
\label{tab:runtime}
\end{table}

\section{Conclusions}
This paper has introduced a general cluster-wise processing network architecture for a fronthaul-limited   CF-mMIMO  system. We  adopt the hybrid  SLINR criterion and proposed two optimization approaches  to maximize the cluster-wise weighted sum pseudo-SE under per-AP transmit power and fronthaul constraints, namely 1) instantaneous CSI-based cluster-wise processing where   precoding, user association, and power allocation are jointly optimized; 2) statistical CSI-based cluster-wise processing where user association and power allocation within a given cluster are jointly optimized. Two modified WMMSE-based algorithms
were  proposed to solve the challenging formulated non-convex mixed-integer problems.  We  investigated the trade-offs provided by the CF-mMIMO system with different number of PCs and highlighted the importance of the appropriate choice of cluster-wise processing relying on either  instantaneous CSI-based or statistical CSI-based design for different system setups. 
Numerical results revealed that  performance loss from increasing the number of processing clusters are primarily influenced by the number of APs, AP antennas, and fronthaul limitations. An interesting observation was that the proposed cluster-wise  processing,  relying only on local CSI, performs fairly close to network-wide alternative that relies on global CSI-knowledge for varying system setups. Investigating the scenarios involving pilot contamination  and the loss of orthogonality between OFDM subcarriers in fronthaul-limited CF-mMIMO systems with cluster-wise processing is recommended for future studies.

\appendices

\bibliographystyle{IEEEtran}
\bibliography{IEEEabrv,Ref_CellFreeFronthaul}

\begin{thebibliography}{10}
\providecommand{\url}[1]{#1}
\csname url@samestyle\endcsname
\providecommand{\newblock}{\relax}
\providecommand{\bibinfo}[2]{#2}
\providecommand{\BIBentrySTDinterwordspacing}{\spaceskip=0pt\relax}
\providecommand{\BIBentryALTinterwordstretchfactor}{4}
\providecommand{\BIBentryALTinterwordspacing}{\spaceskip=\fontdimen2\font plus
\BIBentryALTinterwordstretchfactor\fontdimen3\font minus \fontdimen4\font\relax}
\providecommand{\BIBforeignlanguage}[2]{{%
\expandafter\ifx\csname l@#1\endcsname\relax
\typeout{** WARNING: IEEEtran.bst: No hyphenation pattern has been}%
\typeout{** loaded for the language `#1'. Using the pattern for}%
\typeout{** the default language instead.}%
\else
\language=\csname l@#1\endcsname
\fi
#2}}
\providecommand{\BIBdecl}{\relax}
\BIBdecl

\bibitem{Mohammadali:survey:2024}
M.~Mohammadi, Z.~Mobini, H.~Q. Ngo, and M.~Matthaiou, ``Next-generation multiple access with cell-free massive {MIMO},'' \emph{Proc. {IEEE}}, vol. 112, no.~9, pp. 1372--1420, Sept. 2024.

\bibitem{hien:2017:wcom}
H.~Q. Ngo, A.~Ashikhmin, H.~Yang, E.~G. Larsson, and T.~L. Marzetta, ``Cell-free massive {MIMO} versus small cells,'' \emph{{IEEE} Trans. Wireless Commun.}, vol.~16, no.~3, pp. 1834--1850, Mar. 2017.

\bibitem{Mohammadali:TCOM:2024}
M.~Mohammadi, Z.~Mobini, H.~Q. Ngo, and M.~Matthaiou, ``Ten years of research advances in full-duplex massive {MIMO},'' \emph{{IEEE} Trans. Commun.}, Nov. 2024.

\bibitem{Ngo:PROC:2024}
H.~Q. Ngo, G.~Interdonato, E.~G. Larsson, G.~Caire, and J.~G. Andrews, ``Ultradense cell-free massive {MIMO} for 6{G}: Technical overview and open questions,'' \emph{Proc. {IEEE}}, vol. 112, no.~7, pp. 805--831, July 2024.

\bibitem{Buzzi:WCL:2017}
S.~Buzzi and C.~Dandrea, ``Cell-free massive {MIMO: U}ser-centric approach,'' \emph{IEEE Wireless Commun. Lett.}, vol.~6, no.~6, pp. 706--709, Dec. 2017.

\bibitem{Hien:TGCN:2018}
H.~Q. Ngo, L.-N. Tran, T.~Q. Duong, M.~Matthaiou, and E.~G. Larsson, ``On the total energy efficiency of cell-free massive {MIMO},'' \emph{IEEE Trans. Green Commun. and Networking}, vol.~2, no.~1, pp. 25--39, Mar. 2018.

\bibitem{Ammar:TWC:2021}
H.~A. Ammar, R.~Adve, S.~Shahbazpanahi, G.~Boudreau, and K.~V. Srinivas, ``Downlink resource allocation in multiuser cell-free {MIMO} networks with user-centric clustering,'' \emph{{IEEE} Trans. Wireless Commun.}, vol.~21, no.~3, pp. 1482--1497, Mar. 2022.

\bibitem{Zahra_TWC_Huawei_2025}
Z.~Mobini, A.~H. Gokceoglu, L.~Wang, G.~Peters, and H.~Q. Ngo, ``Fronthaul-aware user-centric generalized cell-free massive {MIMO} systems,'' \emph{{IEEE} Trans. Wireless Commun.}, pp. 1--1, 2025.

\bibitem{Emil:WCOM:2020}
E.~Björnson and L.~Sanguinetti, ``Making cell-free massive {MIMO} competitive with {MMSE} processing and centralized implementation,'' \emph{{IEEE} Trans. Wireless Commun.}, vol.~19, no.~1, pp. 77--90, Jan. 2020.

\bibitem{nayebi:2017:wcom}
E.~Nayebi, A.~Ashikhmin, T.~L. Marzetta, H.~Yang, and B.~D. Rao, ``Precoding and power optimization in cell-free massive {MIMO} systems,'' \emph{{IEEE} Trans. Wireless Commun.}, vol.~16, no.~7, pp. 4445--4459, July 2017.

\bibitem{Pei:2020:TWC}
P.~Liu, K.~Luo, D.~Chen, and T.~Jiang, ``Spectral efficiency analysis of cell-free massive {MIMO} systems with zero-forcing detector,'' \emph{{IEEE} Trans. Wireless Commun.}, vol.~19, no.~2, pp. 795--807, Feb. 2020.

\bibitem{Buzzi:WCOM:2020}
S.~Buzzi, C.~D'Andrea, A.~Zappone, and C.~D'Elia, ``User-centric {5G} cellular networks: {R}esource allocation and comparison with the cell-free massive {MIMO} approach,'' \emph{{IEEE} Trans. Wireless Commun.}, vol.~19, no.~2, pp. 1250--1264, Feb. 2020.

\bibitem{Alonzo:2019:TGCN}
M.~Alonzo, S.~Buzzi, A.~Zappone, and C.~D’Elia, ``Energy-efficient power control in cell-free and user-centric massive {MIMO} at millimeter wave,'' \emph{IEEE Trans. Green Commun. and Networking}, vol.~3, no.~3, pp. 651--663, Sept. 2019.

\bibitem{Mohammadali:JSAC:2023}
M.~Mohammadi, T.~T. Vu, H.~Q. Ngo, and M.~Matthaiou, ``Network-assisted full-duplex cell-free massive {MIMO}: Spectral and energy efficiencies,'' \emph{{IEEE} J. Sel. Areas Commun.}, vol.~41, no.~9, pp. 2833--2851, Sept. 2023.

\bibitem{Hao:IOT:2024}
C.~Hao, T.~T. Vu, H.~Q. Ngo, M.~N. Dao, X.~Dang, C.~Wang, and M.~Matthaiou, ``Joint user association and power control for cell-free massive {MIMO},'' \emph{{IEEE} Internet Things J.}, vol.~11, no.~9, pp. 15\,823--15\,841, May 2024.

\bibitem{Jiafei:WCL:2025}
J.~Fu, Z.~Mobini, H.~Q. Ngo, P.~Zhu, and M.~Matthaiou, ``{WMMSE}-based processing in cell-free massive {MIMO} systems,'' \emph{{IEEE} Wireless Commun. Lett.}, vol.~14, no.~2, pp. 330--334, Feb. 2025.

\bibitem{Atzeni:TWC:2021}
I.~Atzeni, B.~Gouda, and A.~Tölli, ``Distributed precoding design via over-the-air signaling for cell-free massive {MIMO},'' \emph{{IEEE} Trans. Wireless Commun.}, vol.~20, no.~2, pp. 1201--1216, Feb. 2021.

\bibitem{Lorenzo:2022:twc}
L.~Miretti, E.~Björnson, and D.~Gesbert, ``Team {MMSE} precoding with applications to cell-free massive {MIMO},'' \emph{{IEEE} Trans. Wireless Commun.}, vol.~21, no.~8, pp. 6242--6255, Aug. 2022.

\bibitem{Giovanni:2021:SPAWC}
G.~Interdonato and S.~Buzzi, ``Conjugate beamforming with fractional-exponent normalization and scalable power control in cell-free massive {MIMO},'' in \emph{2021 IEEE 22nd International Workshop on Signal Processing Advances in Wireless Communications (SPAWC)}, 2021, pp. 396--400.

\bibitem{Sucharita:2019:Asilomar}
S.~Chakraborty, E.~Björnson, and L.~Sanguinetti, ``Centralized and distributed power allocation for max-min fairness in cell-free massive {MIMO},'' in \emph{2019 53rd Asilomar Conference on Signals, Systems, and Computers}, 2019, pp. 576--580.

\bibitem{Interdonato:TWC:2020}
G.~Interdonato, M.~Karlsson, E.~Björnson, and E.~G. Larsson, ``Local partial zero-forcing precoding for cell-free massive {MIMO},'' \emph{{IEEE} Trans. Wireless Commun.}, vol.~19, no.~7, pp. 4758--4774, July 2020.

\bibitem{Hien:JCOM:2021}
L.~Du, L.~Li, H.~Q. Ngo, T.~C. Mai, and M.~Matthaiou, ``Cell-free massive {MIMO}: Joint maximum-ratio and zero-forcing precoder with power control,'' \emph{{IEEE} Trans. Commun.}, vol.~69, no.~6, pp. 3741--3756, June 2021.

\bibitem{Ammar:TWC:2022}
H.~A. Ammar, R.~Adve, S.~Shahbazpanahi, G.~Boudreau, and K.~V. Srinivas, ``Distributed resource allocation optimization for user-centric cell-free {MIMO} networks,'' \emph{{IEEE} Trans. Wireless Commun.}, vol.~21, no.~5, pp. 3099--3115, May 2022.

\bibitem{Zhang:2024:WCL}
L.~Zhang, S.~Yang, and Z.~Han, ``Pilot assignment for cell-free massive {MIMO}: A spectral clustering approach,'' \emph{{IEEE} Commun. Lett.}, vol.~13, no.~1, pp. 243--247, Jan 2024.

\bibitem{Sadek:2011:TWC}
M.~Sadek and S.~Aissa, ``Leakage based precoding for multi-user {MIMO-OFDM} systems,'' \emph{{IEEE} Trans. Wireless Commun.}, vol.~10, no.~8, pp. 2428--2433, Aug. 2011.

\bibitem{Zahra:2025:LectureNote}
Z.~Mobini and H.~Q. Ngo, ``Massive multiple-input, multiple-output: Instantaneous versus statistical channel state information-based power allocation [lecture notes],'' \emph{{IEEE} Signal Process. Mag.}, vol.~42, no.~2, pp. 27--36, Mar. 2025.

\bibitem{Park:2014:spm}
S.-H. Park, O.~Simeone, O.~Sahin, and S.~Shamai~Shitz, ``Fronthaul compression for cloud radio access networks: Signal processing advances inspired by network information theory,'' \emph{{IEEE} Signal Process. Mag.}, vol.~31, no.~6, pp. 69--79, 2014.

\bibitem{marzetta2016fundamentals}
T.~L. Marzetta and H.~Yang, \emph{Fundamentals of Massive MIMO}.\hskip 1em plus 0.5em minus 0.4em\relax Cambridge, U.K.: Cambridge Univ. Press, 2016.

\bibitem{Peng:TCOM:2023}
Q.~Peng, H.~Ren, C.~Pan, N.~Liu, and M.~Elkashlan, ``Resource allocation for uplink cell-free massive {MIMO} enabled {URLLC} in a smart factory,'' \emph{{IEEE} Trans. Commun.}, vol.~71, no.~1, pp. 553--568, Jan. 2023.

\bibitem{Urlea:2021:JLT}
M.~Urlea and S.~Loyka, ``Simple closed-form approximations for achievable information rates of coded modulation systems,'' \emph{J. Lightw. Technol.}, vol.~39, no.~5, pp. 1306--1311, Mar. 2021.

\bibitem{Binbin:Access:2014}
B.~Dai and W.~Yu, ``Sparse beamforming and user-centric clustering for downlink cloud radio access network,'' \emph{IEEE Access}, vol.~2, pp. 1326--1339, Oct. 2014.

\bibitem{Razaviyayn:ITSP:2013}
Q.~Shi, M.~Razaviyayn, Z.-Q. Luo, and C.~He, ``An iteratively weighted {MMSE} approach to distributed sum-utility maximization for a {MIMO} interfering broadcast channel,'' \emph{{IEEE} Trans. Signal Process.}, vol.~59, no.~9, pp. 4331--4340, Sept. 2011.

\bibitem{Emil2021JCS}
S.~Chakraborty, {\"O}.~T. Demir, E.~Björnson, and P.~Giselsson, ``Efficient downlink power allocation algorithms for cell-free massive {MIMO} systems,'' \emph{IEEE Open J. Commun. Society}, vol.~2, pp. 168--186, Dec. 2021.

\bibitem{Zhou:JSP:2020}
G.~Zhou, C.~Pan, H.~Ren, K.~Wang, and A.~Nallanathan, ``Intelligent reflecting surface aided multigroup multicast {MISO} communication systems,'' \emph{{IEEE} Trans. Signal Process.}, vol.~68, pp. 3236--3251, Apr. 2020.

\end{thebibliography}
\end{document}